\documentclass[11pt]{article}
\usepackage{times}              
\usepackage{latexsym}
\usepackage{amsmath}
\usepackage{amsfonts}
\usepackage{amsthm}
\usepackage{epsfig}
\usepackage{cite}
\usepackage{url}
\usepackage{graphicx}
\usepackage{fancyheadings}
\usepackage{algorithmic}
\usepackage{algorithm}
\usepackage{subfigure}
\usepackage{fullpage}
\usepackage{xcolor,colortbl}
\usepackage{epstopdf}
\usepackage{bbm}
\usepackage{xspace}
\usepackage{hyperref}

\newcommand{\yell}[1]{{\color{red} \textbf{#1}}}
\newcommand{\red}[1]{\textbf{\textcolor{red}{[#1]}}}

\newcommand{\spn}{\mathrm{span}}

\newcommand{\argmin}{\mathop{\textnormal{argmin}}}

\newcommand{\R}{\mathbb{R}}

\newcommand{\pr}{\mathop{\textnormal{Prob}}}

\newcommand{\be}{\begin{enumerate}}
\newcommand{\ee}{\end{enumerate}}
\newcommand{\bi}{\begin{itemize}}
\newcommand{\ei}{\end{itemize}}
\newcommand{\beq}{\begin{equation}}
\newcommand{\eeq}{\end{equation}}

\newcommand{\bp}{\begin{proof}}
\newcommand{\ep}{\end{proof}}
\newcommand{\bcor}{\begin{cor}}
\newcommand{\ecor}{\end{cor}}
\newcommand{\bthm}{\begin{thm}}
\newcommand{\ethm}{\end{thm}}
\newcommand{\blmm}{\begin{lmm}}
\newcommand{\elmm}{\end{lmm}}
\newcommand{\bdefn}{\begin{defn}}
\newcommand{\edefn}{\end{defn}}
\newcommand{\bprop}{\begin{prop}}
\newcommand{\eprop}{\end{prop}}
\newcommand{\bconj}{\begin{conj}}
\newcommand{\econj}{\end{conj}}
\newcommand{\bopm}{\begin{opm}}
\newcommand{\eopm}{\end{opm}}
\newcommand{\brmk}{\begin{rmk}}
\newcommand{\ermk}{\end{rmk}}

\newcommand{\inner}[1]{\left\langle #1 \right\rangle}

\newcommand{\rank}[1]{\mathop{\mathrm{rank}}\left(#1\right)}
\newcommand{\spark}[1]{\mathop{\mathrm{spark}}\left(#1\right)}
\newcommand{\mv}[1]{\mathbf{#1}}

\theoremstyle{plain}                   
\newtheorem{thm}{Theorem}[section]
\newtheorem{lmm}[thm]{Lemma}
\newtheorem{clm}[thm]{Claim}
\newtheorem{prop}[thm]{Proposition}
\newtheorem{cor}[thm]{Corollary}

\theoremstyle{definition}              

\newtheorem{opm}{Open Problem}
\newtheorem{conj}[thm]{Conjecture}

\newtheorem{defn}[thm]{Definition}

\newtheorem{rmk}[thm]{Remark}

\newcommand{\bbox}{
\begin{center}
\begin{tabular}{|c|}
\hline
}
\newcommand{\ebox}{
\\
\hline
\end{tabular}
\end{center}
}

\newlength{\toppush}
\newcommand{\eps}{\epsilon}

\newcommand{\vb}{\mv{b}}
\newcommand{\vA}{\mv{A}}
\newcommand{\Z}{\mathbb{Z}}
\newcommand{\C}{\mathbb{C}}

\newcommand{\iconst}{{\rm i}}
\newcommand{\Kerdock}{\mathcal{K}}
\newcommand{\inabs}[1]{\left|#1\right|}
\newcommand{\vbeta}{{\boldsymbol\beta}}
\newcommand{\mspn}[1]{\mathrm{span}(#1)}
\newcommand{\epsbound}{\sqrt{1-17\mu}}

\newcommand{\wand}{\quad\textnormal{\textbf{and}}\quad}
\newcommand{\wor}{\quad\textnormal{\textbf{or}}\quad}

\newcommand{\sparse}{\textsc{Sparse}\xspace}
\newcommand{\listsparse}{\textsc{List-Sparse}\xspace}
\newcommand{\listapprox}{\textsc{List-Approx}\xspace}

\newcommand{\cL}{\mathcal{L}}

\newcommand{\msupp}[1]{\mathrm{supp}\left(#1\right)}

\title{\textbf{Sparse Approximation, List Decoding, and Uncertainty Principles}\footnote{MAK, HQN and AR's research is partially supported by NSF grant CCF-1161196. ACG's research is partially supported by NSF grant CCF-1161233.}}
\author{\textsc{Mahmoud Abo Khamis}\footnotemark[2]
\and
\textsc{Anna C.~Gilbert}\footnotemark[3]
\and
\textsc{Hung Q.~Ngo}\footnotemark[2]
\and
\textsc{Atri Rudra}\footnotemark[2]
}

\date{\today\\
\vspace*{4mm}
\footnotemark[2]~ Department of Computer Science and Engineering,\\
University at Buffalo, SUNY.\\
\texttt{\{mabokham,hungngo,atri\}@buffalo.edu}\\
\vspace*{2mm}
\footnotemark[3]~ Department of Mathematics,\\
University of Michigan.\\
\texttt{annacg@umich.edu}}

\begin{document}

\maketitle

\thispagestyle{empty}

\begin{abstract}
We consider list versions of sparse approximation problems, where unlike the existing results in sparse approximation that consider situations with unique solutions, we are interested in multiple solutions. We introduce these problems and present the first combinatorial results on the output list size. These generalize and enhance some of the existing results on threshold phenomenon and uncertainty principles in sparse approximations. Our definitions and results are inspired by similar results in list decoding. We also present lower bound examples that bolster our results and show they are of the appropriate size.

\end{abstract}

\newpage


\section{Introduction} \label{sec:intro}
One of the fundamental mathematical problems in modern signal processing, data compression, dimension reduction for large data sets, and streaming algorithms is the \emph{sparse approximation problem}: Given a matrix or a redundant dictionary $\mv A \in \R^{m \times N}$ and a vector $\mv b \in \R^m$, find
\begin{equation}
\label{eqn:exact}
	\argmin \|\mv x\|_0 \quad\text{such that}\quad \mv A \mv x = \mv b.
\end{equation}
In other words, find the sparsest $\mv x \in \R^N$ (i.e., the vector with the fewest non-zero entries) such that $\mv A \mv x = \mv b$ or that represents $\mv b$ exactly. In general, this problem is NP-complete. There are several variations of this problem, including $k$-sparse approximation
\begin{equation}
	\label{eqn:k-sparse}
	\min \| \mv A \mv x - \mv b\|_2 \quad\text{such that}\quad \|\mv x\|_0 \leq k
\end{equation}
(i.e., find the $k$-sparse vector $\mv x$ that minimizes $\|\mv A \mv x - \mv b\|_2$), as well as the convex relaxation of Equation~\eqref{eqn:exact}
\begin{equation}
	\label{eqn:l1min}
	\argmin \|\mv x\|_1 \quad\text{such that}\quad \mv A \mv x = \mv b.
\end{equation}
All of these variations capture different aspects of the applications of sparse approximation in signal processing, streaming algorithms, etc.\footnote{We remark that these problems are related to the sparse recovery/compressed sensing but differ in the way in which the vector $\mv x$ is evaluated and the conditions on $\mv A$. We refer the reader to the survey by Gilbert and Indyk~\cite{cs-survey} for more details.} 

All of these problems exhibit fundamental trade-offs and threshold phenomena. We highlight two such examples: (i) a phase transition in the convex optimization in Equation~\eqref{eqn:l1min} and (ii) a simple relationship between the sparsity $k$ of $\mv x$ and the coherence $\mu$ of $\mv A$ (the maximum, in absolute value,  dot-product between any two pairs of columns) when $\mv A$ is the union of two orthonormal bases for $\R^m$ that ensures the uniqueness of the solution to Equation~\eqref{eqn:exact}. Donoho and Tanner~\cite{DonohoTanner:Thresholds2006,DonohoTanner:Neighborliness2005} first observed that if we define the convex program successful when it returns the unique optimal vector $\hat{\mv x}$ that equals the true unknown solution $\mv x$, then the observed probability of success (taken over the random choice of $\mv A$) exhibits a sharp phase transition as the sparsity ratio $k/N$ and the redundancy ratio $m/N$ range from 0 to 1. Many other authors have analyzed and demonstrated empirically this phase transition in a variety of other convex programs, including Amelunxen et al.~\cite{Tropp:Edge2013} who provide a geometric theory of this ubiquitous phenomenon. However, these results are not applicable to the original sparse approximation problem \eqref{eqn:k-sparse}, which is not a convex optimization problem.

The second trade-off or threshold phenomenon was first realized by Donoho and Huo~\cite{DonohoHuo:UP2001} and then expanded upon in a series of papers by Bruckstein, Elad, and others~\cite{DonohoElad:2003,GribonvalNielsen:2007,EladBruckstein:UPPairs2002,BrucksteinDonohoElad:Review2009}. In its simplest form, we assume that $\mv A = [\Phi, \Psi]$ is the union of two orthonormal bases in $\R^m$ and that the coherence 
 of $\mv A$ is $\mu$.
\begin{thm}~\cite{EladBruckstein:UPPairs2002}
	Assuming
\begin{equation}
	\label{eqn:cohbound}
	k < \frac{1}{\mu},
\end{equation}
the exact sparse representation problem~\eqref{eqn:exact} has a unique solution.
\end{thm} 
This result is proven using a (discrete) \emph{uncertainty principle}. Suppose that $\mv b$ has two different representations in each basis $\Phi$ and $\Psi$ and that each representation has sparsity $k_{\Phi}$ and $k_{\Psi}$, respectively, then
$
	k_{\Phi} + k_{\Psi} \geq \frac{2}{\mu}$.

Despite the intense scrutiny of these thresholds, scant attention has been paid to solutions of Equation~\eqref{eqn:exact} beyond them, nor is there a rigorous analysis of combinatorial bounds for Equation~\eqref{eqn:k-sparse}. We highlight three relevant results. In~\cite{Tropp:SparsityGap2010}, Tropp describes the gap in sparsity between two non-unique exact representations of a random vector $\mv b$ over a matrix $\mv A$ with uniform coherence $\mu$. 
\begin{thm}~\cite{Tropp:SparsityGap2010}
\label{thm:Tropp-Sparsity-Gap}
	We draw a support set $S$ of linearly independent columns of $\mv A$ and set $\mv b$ equal to a random linear combination of these columns. Then, almost surely, we cannot represent $\mv b$ with a set $T$, disjoint from $S$, unless
	$
		|T| > |S| \Big( \frac{\mu^{-1}}{\sqrt{|S|}} - 1  \Big)$.
	
\end{thm}
Dragotti and Lu~\cite{DragottiLu:ProSparse2013} begin with a specific pair of orthonormal bases, the canonical basis and the Fourier basis (referred to informally as ``spikes and sines''), and construct a polynomial time algorithm that can return a list\footnote{However, the analysis of the list size is not fully fleshed out.} of exact $k$-term representations, assuming $k < \sqrt{2}/\mu$.

Finally, Donoho and Elad~\cite{DonohoElad:2003} give a general uniqueness result for the exact problem~\eqref{eqn:exact} that uses the spark\footnote{The spark of a matrix is the minimum number of dependent columns in the matrix.}  of $\mv A$ and that is considerably more general than~\eqref{eqn:cohbound}.

\paragraph{Our Contributions.} All of these results above consider the regime where one is interested in a unique solution. In particular, existing work has stopped at this threshold of unique solutions.
We propose a rigorous interpretation and analysis of sparse approximation beyond these thresholds, using list decoding of error correcting codes as an analogy. List decoding was proposed by Elias~\cite{elias} and Wozencraft~\cite{wozencraft} in coding theory (for discrete alphabets and Hamming distance). List decoding is a relaxation of the usual ``unique decoding" paradigm, where the decoder is allowed to output a small list of codewords with the guarantee that the transmitted codeword is in the list. (Unique decoding is the special case where the list can only have one codeword.) This allows one to go beyond the well-known half the distance bound for the number of correctable worst-case errors with unique decoding. In many situations, list decoding allows for correcting twice as many errors as one can with unique decoding. In particular, one can correct close to $100\%$ of errors (as opposed to the $50\%$ bound for unique decoding) with constant sized lists. This remarkable fact found many surprising applications in complexity theory (see e.g. the survey by Sudan~\cite{madhu-survey} and Guruswami's thesis~\cite{venkat-thesis} for more on these connections). Motivated by these applications, much progress has been made in algorithmic aspects of list decoding (for more details see~\cite{book}).

We now briefly focus on combinatorial list decoding. In particular, we are interested in sufficient conditions that allow codes to have small output list sizes. Perhaps the most general such result in literature is the Johnson bound, which states that a code with relative distance $\delta$ can be list decoded to $1-\sqrt{1-\delta}$ fraction of errors with small list size~\cite{johnson-1,johnson-2}. Note that as $\delta\to 1$, the fraction of correctable error approaches $100\%$. As the natural analogue of distance for sparse approximation is coherence, one would hope to prove similar results in list sparse approximation. In this work, we do so. 

Informally, we define \emph{list sparse approximation} as the problem of returning a list of $k$-term representations $\mv x$ such that $\mv A \mv x = \mv b$ or $\|\mv A \mv x - \mv b\|_2$ is minimized. Just as list decoding is a more flexible definition of decoding, list sparse approximation is a more flexible notion of representing a vector sparsely over a redundant dictionary and one that, we will show, permits us to move ``beyond'' the traditional bounds for sparse approximation. More formally, we define two variations of list sparse approximation: \textsc{List-Approx} and \textsc{List-Sparse}. The differences among the variations cover whether we require the representations in the list  only to \emph{approximate} or to \emph{equal} the input vector $\mv b$.

Several results addressed the \emph{uniqueness} of \emph{exact} sparse representations. A collection of such results can be found in \cite{BrucksteinDonohoElad:Review2009}. Just like unique decoding is a special case of list decoding (where list size is 1) and exact representation is a special case of approximate representation (where error tolerance is zero), we will show that some of those results are special cases of our results. Our results extend those results in two directions: \emph{approximation} and \emph{multiplicity} of solutions.

We believe that the list sparse approximation questions that we study are interesting mathematically in their own right. In addition, we are hopeful that this notion of list sparse approximation will find other theoretical applications (just as the notion of list decoding found many applications). In general, this notion should find applications in situations where, in addition to computing an approximation close to the vector $\vb$, one also has a secondary objective. In such a scenario, having a list of approximations that are the same with respect to error of approximation and sparsity would be beneficial. We leave the problem of finding other applications of list sparse approximation where it strictly outperforms the traditional unique sparse approximation as a tantalizing open question from our work.

We illustrate the potential application of these concepts with a simple image compression example.  We use the Haar wavelet packet redundant dictionary to define three different classes of $k$-sparse approximations, each of which compresses the image to one fifth its original size with a relative error of no more than 0.01; i.e., we provide examples of three different solutions to the \textsc{List-Approx} problem.  Figure~\ref{fig:blobs} shows the original image and its three different list sparse approximations. We can see that one of the approximations (Class 1) is not quite as accurate as the other two although it is easier to compute than the others. Also Class 2 and 3 are results of trying to optimize different auxiliary objectives. Thus, this illustrates a scenario where having a list of solutions that contains all three classes could be more beneficial than just trying to output a single solution that matches the bounds on sparsity $k$ and error $\eps$. Details of this particular computation can be found in the Appendix~\ref{sec:image_appendix}.
\begin{figure}[ht!]
		\begin{center}
			\includegraphics[height=2in,width=2.75in]{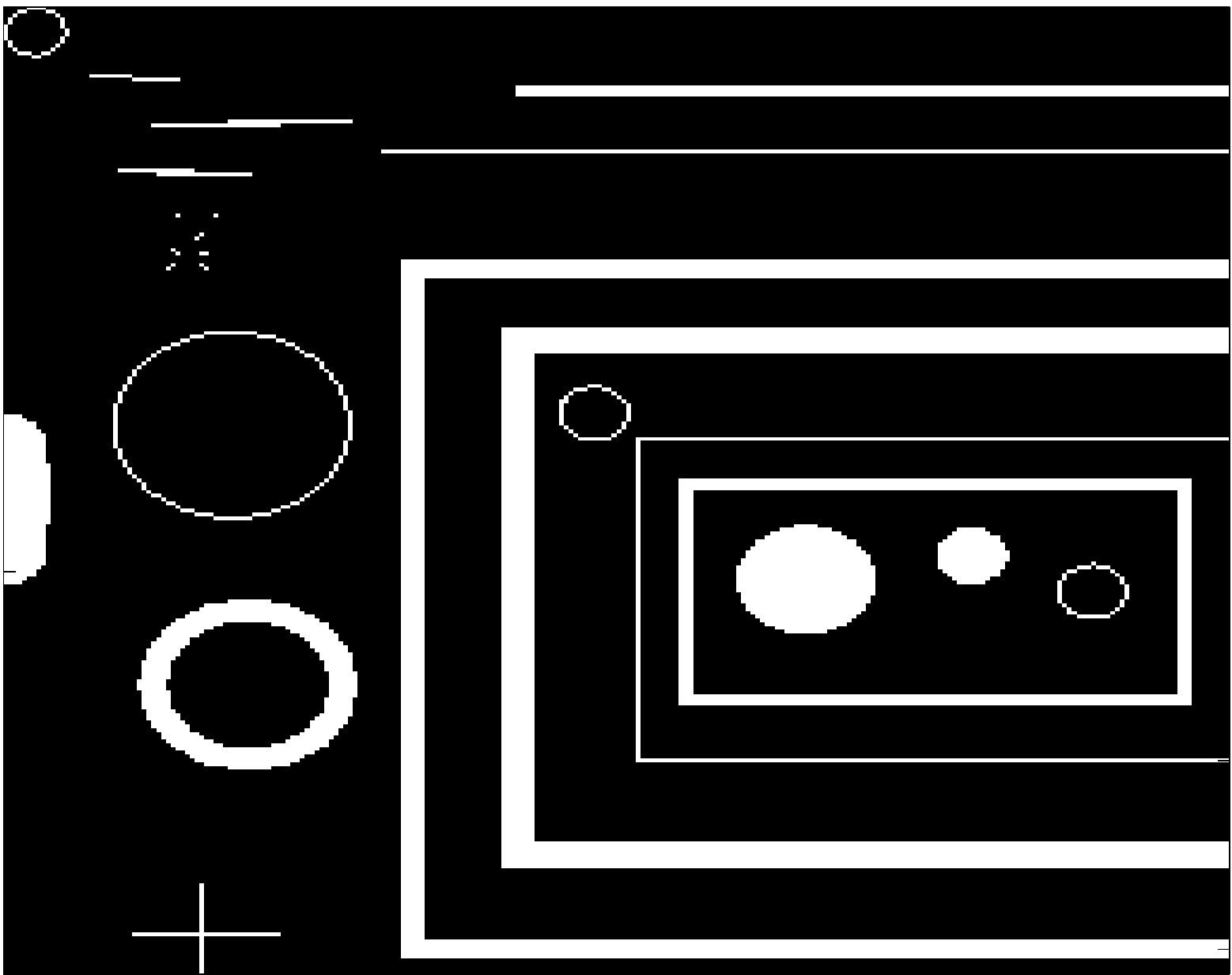}
			\includegraphics[height=2in,width=2.75in]{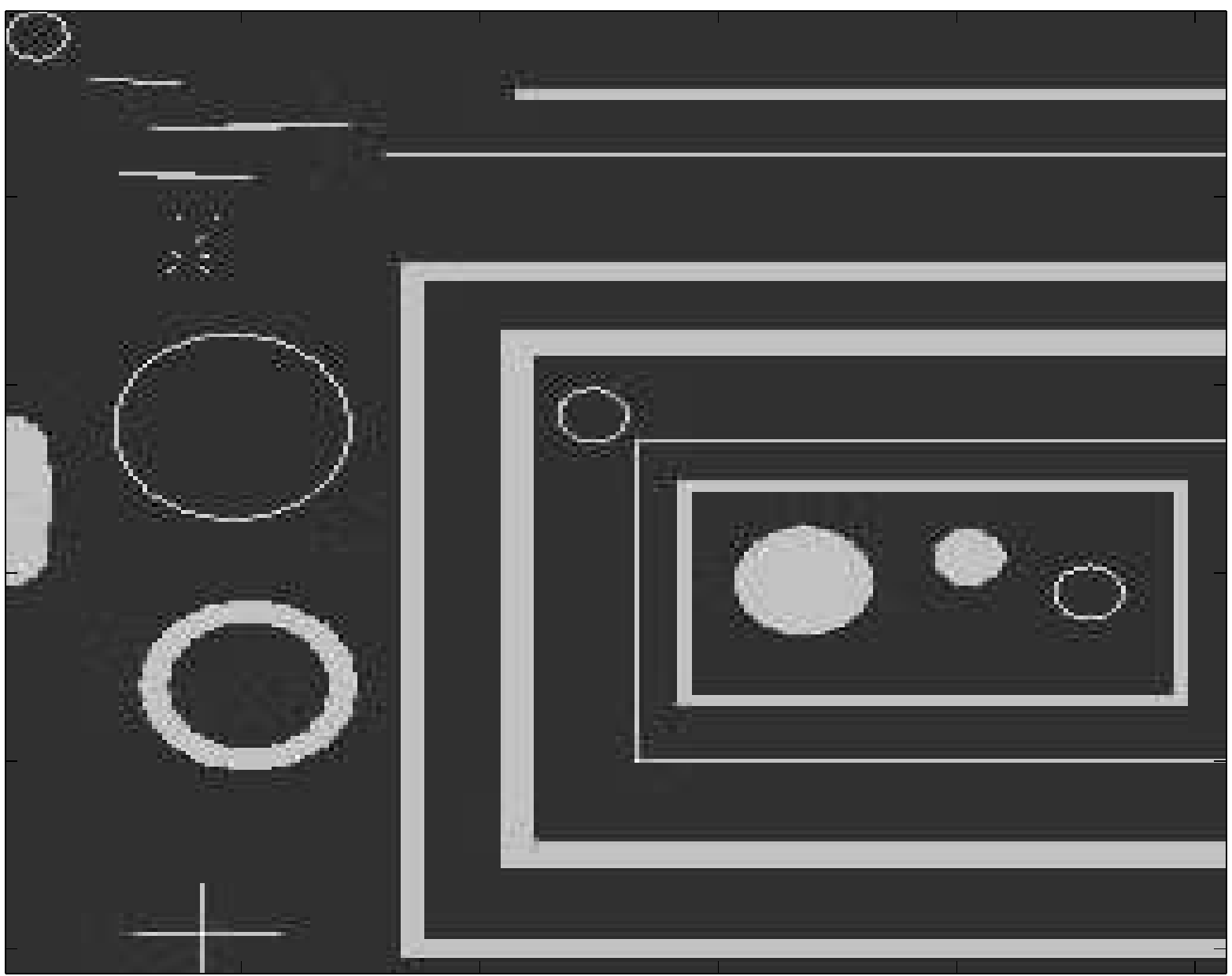}
			\includegraphics[height=2in,width=2.75in]{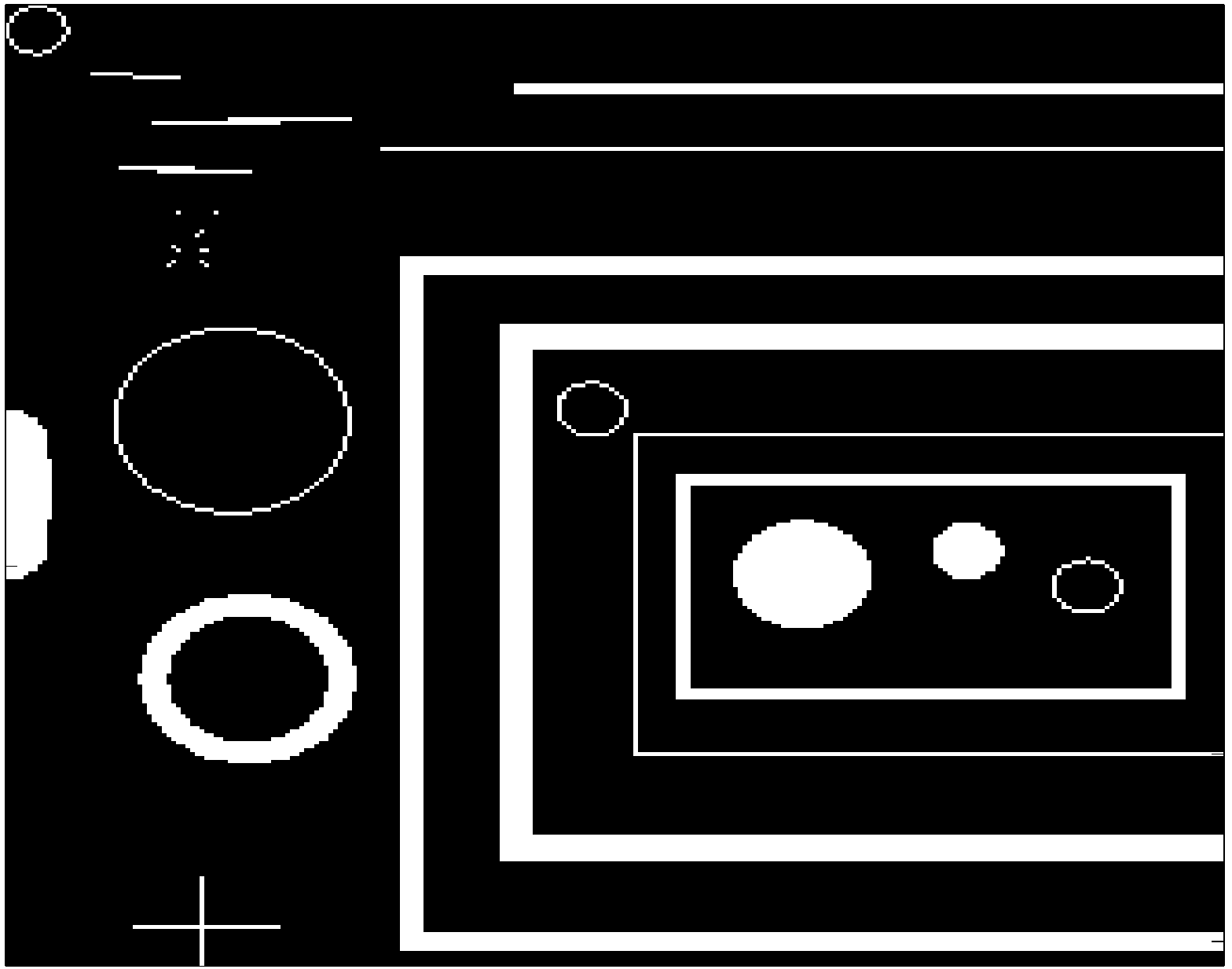}
			\includegraphics[height=2in,width=2.75in]{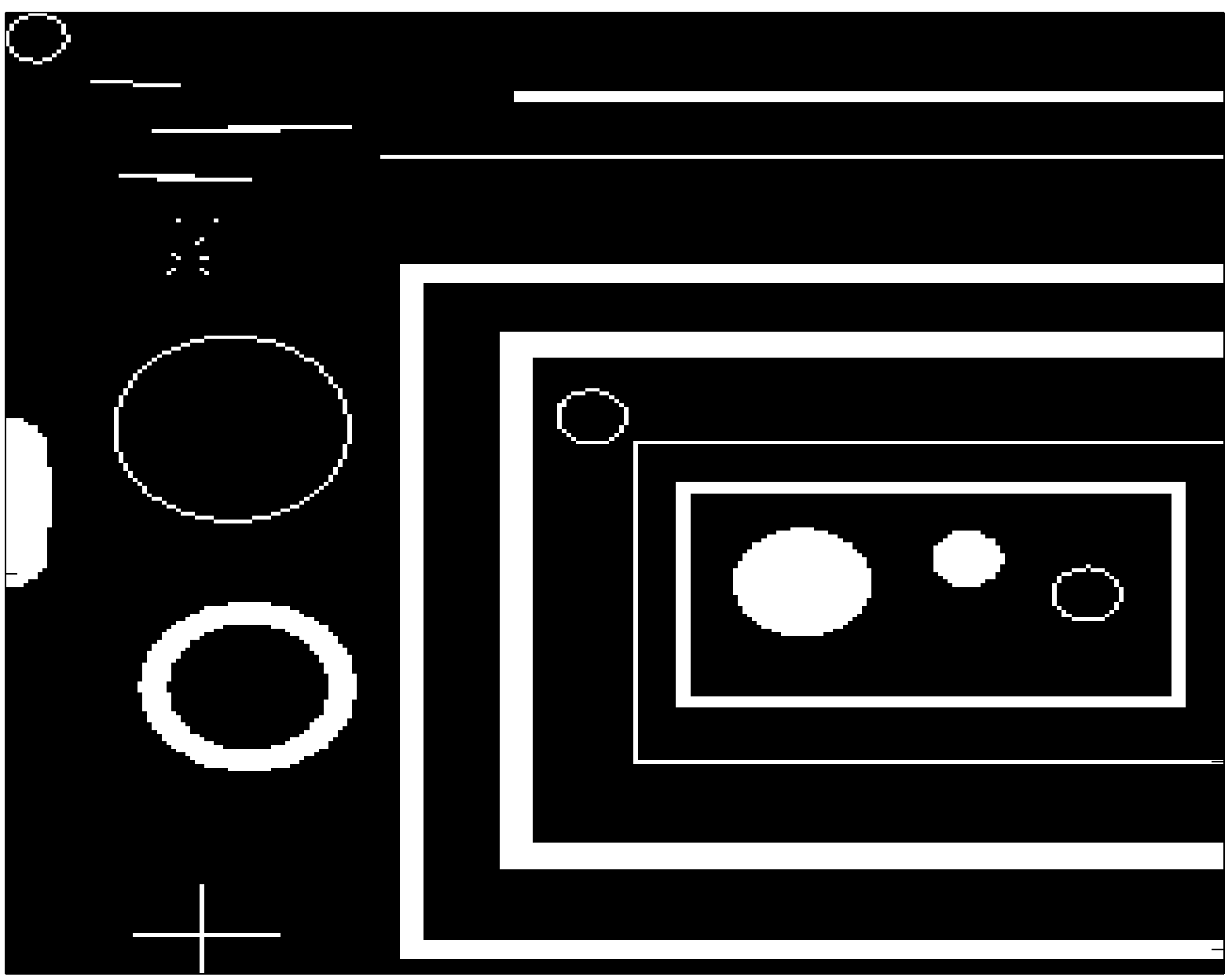}
		\end{center}
		\caption{(Upper left.) Original image from the Matlab image processing toolbox demo images, \texttt{blobs}. (Upper right.) Class 1: Large and random medium coefficients. (Lower left.) Class 2: Truncated BestBasis, computed minimizing the entropy of the $\ell_2$ norm of the coefficients. (Lower right.) Class 3: Truncated BestBasis, computed minimizing the $\ell_1$ norm of the coefficients.}
		\label{fig:blobs}
\end{figure}

\paragraph{Our Results.}

We relax our definition of sparse representation to include those $k$-term representations that are within a specified distance $\epsilon$ of the input vector $\mv b$ in \textsc{List-Approx} and find that:
\begin{itemize}
	\item As long as the distance $\eps\le \sqrt{1-\Omega(\mu k)}$, then we have that the number of disjoint solutions is $O(1/(1-\eps^2))$. In fact, we show that if we only consider solutions where each atom appears only $o(L)$ times in the output list of size $L$, then $L$ is bounded only by $\eps$ (and is independent of $k$ and $N$). These results extend the current uncertainty principles in three ways: (i) We now consider the case when the approximation error $\eps>0$ (all existing work considers the exact case of $\eps=0$), (ii) We consider the case of larger multiplicity of disjoint solutions, i.e. more than two disjoint solutions (where uncertainty principle results only consider the case of two disjoint solutions) and (iii) We consider the generalization where solutions can have (limited) overlap (this scenario was not considered in Theorem~\ref{thm:Tropp-Sparsity-Gap}). Further, we obtain as simple corollaries all the known conditions that guarantee unique decoding; i.e., list size of $1$---see Section~\ref{sec:recover-old-results} for more details. 
	\item We show that to obtain a list size that only depends on $\eps$, our bound is essentially optimal in terms of the dependence of $\mu$ on $k$.
\end{itemize}

To fully understand the implication of these upper bounds and to formulate lower bounds, we construct several examples, including for the spikes and sines dictionary (the proto-dictionary for which we have an uncertainty principle) and the dictionary obtained from Kerdock codes that exhibit a large exponential list size. 

In addition, we establish the following basic combinatorial bounds on the list size of \textsc{List-Sparse} (where we want to output all $\mathbf x$ such that $\|\mathbf A\mathbf x-\mathbf b\|_2$ is minimized):
\begin{itemize}
	\item Given a matrix $\mv A$, we determine necessary and sufficient conditions to ensure that for all input vectors $\mv b$, the list of $k$-sparse representations is finite as well as trade-offs to ensure the list is no longer than a specified length $L$.
	\item We calculate the minimum worst-case $k$-sparse list size over all matrices $\mv A$; it is $k+1$.
	\item Given a matrix $\mv A$, we determine necessary and sufficient conditions to ensure that most input vectors $\mv b$ generate a finite list or a list of size 1.
\end{itemize}
To the best of our knowledge, this is the first work that considers the size of the solution space of the \textsc{Sparse} problem.
 

\paragraph{Our Techniques.} Our results follow from fairly simple arguments. For the \textsc{List-Approx} problem we have two sets of results. In the first set of results that allow for every atom to be present in many solutions, we work with the average error of the \textsc{List-Approx} version instead of the more natural ``maximum error" in our proofs. The former implies the latter so our upper bounds are also valid for the maximum error version. A similar switch has been explicitly used fairly recently in the list decoding literature~\cite{gurnar2013,woot2013,rw14}, though its implicit use dates back to at least the Johnson bound~\cite{johnson-1,johnson-2}. Our second set of results is for the case when the solutions are disjoint (for which we get tighter bounds) where we use the fact that a simplex with $L$ vertices forms the worst-case setting for lists of size $L$. This is unlike the Hamming setting, where there is no natural analogue for $L>2$.

Our results on the \textsc{List-Sparse} problem essentially follow from the observation that the columns of the matrix $\mv A$ divide up the space into Voronoi cells and the vectors $\mv b$ that give rise to many solutions are those that lie at the boundary of many of these Voronoi cells.


\paragraph{Organization of the paper.} We start with some preliminaries in Section~\ref{sec:prelims} where we formally define the problems of \textsc{List-Approx} and \textsc{List-Sparse}. We present our results on \textsc{List-Approx} in Section~\ref{sec:listapprox} and our results on \textsc{List-Sparse} in Section~\ref{sec:listsparse}. We present examples that show the tightness of some of the aspects of our upper bounds in Section~\ref{sec:examples}. We conclude with some open questions in Section~\ref{sec:concl}. To facilitate the ease of reading, all proofs are deferred to the appendix.

\section{Preliminaries}
\label{sec:prelims}

Given the close connections between coding theory and sparse approximation, it is natural to ask whether we can extend the current combinatorial results on sparse approximation beyond the spark bound using a notion of \emph{list sparse approximation}. Just as in error correcting codes list decoding enables us to decode a much greater fraction of errors, we anticipate list sparse approximation to approximate vectors with much less sparsity. To that end, we study two natural extensions to the \sparse problem.

First, let us recall the ingredients. 
\begin{defn}
	A redundant dictionary $\mv{A}$ is a matrix of size $m \times N$ with $N \geq m$, the columns of $\mv{A}$ span $\R^m$, and the columns, which we refer to as \emph{atoms}, are normalized to have $\ell_2$ or Euclidean norm 1. Because the atoms span $\R^m$, the rank of $\mv{A}$ is $m$.
\end{defn}

\sparse seeks the best (linear) representation of an input vector $\mv{b}$ over a redundant dictionary $\mv{A}$ that uses at most $k$ atoms.
\begin{defn}[\sparse]
	Given $\mv{A}$ and $\mv{b}$, find
	\beq
		\hat{ \mv{x}} = \argmin_{\|\mv {x}\|_0 \leq k } \| \mv{Ax-b} \|_2.
		\label{eq:Sparse-problem}
	\eeq
\end{defn}
A vector $\mv x$ for which $\|\mv x\|_0 \leq k$ is called a {\em $k$-sparse} vector. Note that a solution to \sparse need not be an exact representation for $\mv b$. The first extension merely lists all optimal solutions to \sparse.
\begin{defn}[\listsparse]
	Given $\mv{A}$, $\mv{b}$, and $k$, list \emph{all} $k$-sparse $\mv{x}$ such that $\| \mv{Ax-b} \|_2$ is minimized. 
\end{defn}
The second extension relaxes the instance on optimal solutions and replaces it with listing all $k$-sparse vectors with error no more than $\epsilon$, a parameter.
\begin{defn}[\listapprox]
\label{defn:listapprox}
	Given $\mv{A}$, $\mv{b}$, $k$, and $\epsilon$, list \emph{all} $k$-sparse $\mv{x}$ such that
	$
		\| \mv{Ax - b} \|_2 \leq \epsilon$.
\end{defn}
(There is a catch in how we define ``all" solutions for the \listapprox problem. See Definition~\ref{def:L} for more on this.)
There is a third possible extension that relaxes the instance on optimal solutions in a different fashion and requires that solutions to be close to the optimal error up to an \emph{additive} tolerance $\epsilon$.
	Given $\mv{A}$, $\mv{b}$, $k$, and $\epsilon$, list \emph{all} $k$-sparse $\mv{x}$ such that
		$\| \mv{Ax - b} \|_2 \leq \min_{\| \mv{y} \|_0 \leq k} \| \mv{Ay - b} \|_2 + \epsilon$.
Observe that if we define the optimal error 
$
	\epsilon_{\rm{opt}} = \min_{\| \mv{y} \|_0 \leq k} \| \mv{Ay - b} \|_2$ 
and seek the solution to the third variant  with additive error $\epsilon'$, it is sufficient to solve \listapprox with error $\epsilon = \epsilon_{\rm{opt}} + \epsilon'$. For this reason, we focus our attention on \listapprox.

Let us assume that the input vector $\mv b$ is normalized to have unit norm as well. In this case, \listapprox is the direct analogy in sparse approximation with list decoding. To see this, let us recall the basic definitions in list decoding. The alphabet $\Sigma$ is a finite set and the encoding function is a map $E: \Sigma^k \rightarrow \Sigma^n$. Given a received word $\mv r \in \Sigma^n$, the goal is to return a list of messages $\mv m\in\Sigma^k$ such that the Hamming distance between $\mv r$ and $E(\mv m)$ is at most $e$, the number of errors. It is clear that if we set $k = 1$, use a finite alphabet $\Sigma$ (rather than $\R$), let $\mv{A}$ be the codebook, $\mv{b}$ be the received codeword, $\epsilon$ the error, and convert the $\ell_2$ metric to (relative) Hamming distance, then we have the list decoding problem.

In keeping with the analogy between error-correcting codes and sparse
approximation, we recall the definition of the \emph{coherence} of a redundant
dictionary, an analogous quantity to the inverse of the \emph{distance} of an error-correcting
code (as the distance between vectors decreases, their coherence increases). 
\begin{defn}[Coherence]
	The \emph{coherence} of a redundant dictionary $\mv A$ is the largest (in absolute value) dot product between any two atoms in the dictionary:
	\[
		\mu(\mv A) = \max_{i \neq j} | \langle \mv A_i, \mv A_j \rangle |.
	\]
\end{defn}

The following uniqueness result for the exact problem uses the coherence of $\mv A$ and appears in \cite{BrucksteinDonohoElad:Review2009}:

\begin{thm}
\label{thm:uniqueness-via-mu}
	A representation $\mv x$ is the unique solution to \eqref{eqn:exact} if its sparsity 
			$k < \frac{1}{2}(\mu(\mv A)^{-1} +1)$.
\end{thm}

In much of our analysis of \listsparse, we use the
\emph{spark}~\cite{DonohoElad:2003} of the redundant dictionary as a more
refined geometric property than the coherence. Spark is a measure of the linear dependence among the columns of $\mv A$ but one that is considerably different from the rank.
\begin{defn}[Spark]
	The \emph{spark} of a redundant dictionary $\mv{A}$ is the smallest $s$ such
    that some $s$ columns of $\mv{A}$ are linearly dependent:
	\[
		\spark{\mv A} = \min_{\mv z \neq 0} \|\mv z\|_0 \quad\text{s.t.}\quad \mv A\mv z = \mv 0.
	\]
\end{defn}

Donoho and Elad~\cite{DonohoElad:2003,BrucksteinDonohoElad:Review2009} give a general uniqueness result for the exact problem~\eqref{eqn:exact} that uses the spark of $\mv A$ and that is considerably more general than~\eqref{eqn:cohbound}:
\begin{thm}
\label{thm:uniqueness-via-spark}
	A representation $\mv x$ is the unique solution to \eqref{eqn:exact} if its sparsity 
			$k < \spark{\mv A}/2$.
\end{thm}
For $\mv A = [\Phi, \Psi]$ the union of two orthonormal bases with coherence $\mu$, $\spark{\mv A} \geq 2/\mu$, so this theorem not only implies \eqref{eqn:cohbound} but also the uncertainty principle. Both Theorems~\ref{thm:uniqueness-via-spark} and \ref{thm:uniqueness-via-mu} are implied by our results, as we will show in Section~\ref{sec:recover-old-results}.

In the following sections, we aim to bound the size of the list of solutions to \listsparse and to \listapprox. For \listapprox, counting the number of sets on which these solutions are supported is more meaningful as there may be an infinite number of different solutions supported on any one set of columns of $\mv A$. That is, we treat the \listapprox problem as a combinatorial one rather than an optimization problem. 
To that end, we define the list size.
\begin{defn}
For a given matrix $\mv A$ of size $m \times N$, a vector $\mv b$ of length $m$,
and an integer $1\leq k \leq N$, let $s(\mv A, \mv b, k)$ be the number of
optimal solutions $\mv x$ to \sparse (Problem \eqref{eq:Sparse-problem}).  
The quantity $s(\mv A, \mv b, k)$ is the {\em list size} in the list decoding
sense.
\end{defn}

For \listapprox, we define the list size as follows.
\begin{defn}
\label{def:L}
	Given $\mv{A}$, $\mv{b}$, $k$, and $\epsilon$, let $L(\mv{A},\mv{b},k,\eps)$ be the number of distinct {\em supports} of solutions $\mv{x}$ with sparsity $k$ that satisfy
	$
		\|\mv{Ax - b}\|_2 \leq \epsilon$.
We will also use $L(\mv{A},k,\eps,R)$ to denote the worst case bound on $L$ over all $\mv{b}$ with the restriction that no atom appears in the support of more than $R$ out of the $L$ solutions.
\end{defn}

\section{\listapprox}
\label{sec:listapprox}

In this section, we focus on the \listapprox problem as it is the direct analogue to list decoding for sparse approximation. 
	We would like to bound the quantity $L(\mv{A, b},k,\eps)$
in terms of the coherence $\mu$ of $\mv{A}$. Intuitively, the smaller the $\mu$ the smaller the quantity $L(\mv{A},\vb,k,\eps)$. However, it is not too hard to see that $\max_{\mv b} L(\mv{A},\vb,k,\eps)$ can be as high as $m^{\Omega(k)}$ even if $\mu=0$ for $k>1$. 

\begin{prop}
\label{prop:mu=0-exp-lb}
There exists a matrix $\mv{A}$ with coherence $\mu=0$ such that for every $k\ge 1$ and 
\begin{equation}
\label{eq:eps-bound}
\eps<\sqrt{\frac{m-k}{m}},
\end{equation}
there exists a $\mv b$ with 
$ L(\mv{A},\vb,k,\eps) \ge \binom{m-1}{k-1}$.
\end{prop}

This implies that in general one cannot hope to have a non-trivial bound on $L(\mv{A},\vb,k,\eps)$ in terms of the coherence of $\mv{A}$. If one looks at the bad example in Proposition~\ref{prop:mu=0-exp-lb}, then one observes that the ``culprit" was one atom that appeared in all of the solutions. A natural way to avoid this bad case would be to only consider a list of solutions, where no atom appears in ``too many" solutions. More precisely, we will consider the case where there can be up to $L$ solutions (in terms of the support sets) to $\|\mv{A}\mv{x}-\mv{b}\|_2\le \eps$ that are $k$-sparse such that no atom in $\mv{A}$ occurs in the support of more than $o(L)$ solutions. 
We will use $L(\mv{A},k,\eps,R)$ to denote the worst case bound on $L$ over all $\mv{b}$ with the restriction that no atom appears in the support of more than $R$ out of the $L$ solutions.

\subsection{List size bound $L(\mv{A},k,\eps,o(L)$)}

We prove the following result. 

\begin{prop}
\label{prop:av-p2-k}
Let $0<\gamma<1$ be a real number and $k>1$ be an integer.
Assume that the coherence of the dictionary $\mv{A}$ is $\mu$. Then as long as
\begin{equation}
\label{eq:k-mu-lb-org}
\eps \le \sqrt{1-24\cdot \left( \mu k\right)^{1-\gamma}},
\end{equation}
we have
$L(\mv{A},k,\eps,L^\gamma)\le \left\lceil \left(\frac{11}{1-\eps^2}\right)^{1/(1-\gamma)}\right\rceil$.
\end{prop}

The following result immediately follows from the proof of Proposition~\ref{prop:av-p2-k}:
\begin{thm}
\label{thm:gen-listapprox-ub}
Let  $k>1$ be an integer.
Assume that the coherence of the dictionary $\mv{A}$ is $\mu$. Then as long as
$\mu \le \frac{1}{2kL}$,
we have 
$L=L(\mv{A},k,\eps,o(L))$ is  $O_{\eps}(1)$
i.e. $L(\mv{A},k,\eps,o(L))$ is a constant that just depends on $\eps$ (and how far the function $o(L)$ is from $L$).
\end{thm}

We note in Appendix~\ref{app:kerdock} that the bound of $\mu\le O(1/k)$ in the above result to get a constant list size is necessary. 

\subsection{List size bound $L(\mv{A},k,\eps, 1)$}

In this section, we present sharper bounds on $L(\mv{A},k,\eps, 1)$ than those presented in the previous section.
As we mentioned earlier, \listapprox can be viewed as a list-decoding problem. In fact, \listapprox with $k=1$ is the list-decoding problem of spherical codes. Here, we aim to bound the list size $L(\mv{A},k,\eps,1)$ when no atom appears in more than one solution (i.e., $R=1$). First, we will develop a bound for spherical code list-decoding. Then, we will generalize it to $k>1$. It might be useful to consider Euclidean codes to build some intuition-- see Appendix~\ref{app:euclid} for more details.

\subsubsection{$k=1$: A list-decoding bound in spherical codes}
\bdefn[Spherical code]
In a spherical code, the codewords are unit vectors that are the columns of a dictionary $\mv A$.
\edefn

Given a dictionary $\mv A$ whose coherence is $\mu(\mv A)$, a unit-length target vector $\mv b$, and an error bound $\epsilon$; the list decoding problem is to output a list of all columns $\mv A_i$ that satisfy the following:
\begin{equation}
\label{eq:spherical-list-decoding}
\min_{x\in \R}\| \mv A_i x -\mv b \|_2 \leq \epsilon.
\end{equation}
Notice that list-decoding of spherical codes corresponds to \listapprox with $k=1$. 

\bthm[Bound on list-decoding of spherical codes]
\label{thm:list-decoding-spherical}
Given a spherical code represented by a dictionary $\mv A$ (whose coherence is $\mu(\mv A)$) and an error bound $\epsilon$; if $\eps<\sqrt{1-\mu(\mv A)}$, then the maximum list size $L(\mv A,1,\eps,1)$ is bounded by
\begin{equation}
\label{eq:list-decoding-spherical}
L(\mv A,1,\eps,1)\leq \left\lfloor\frac{1}{1-\frac{\epsilon^2}{1-\mu(\mv A)}}\right\rfloor.
\end{equation}
Moreover, the above bound is tight if the right-hand side is $\leq m$, where $m$ is the dimension of the code. 
\ethm

\subsubsection{$L(\mv A, k,\eps,1)$ for $k>1$}
First, we need to define a new concept.
Recall that the coherence of a dictionary is the cosine of the minimum angle between any two columns.
Given two subsets of columns, we can take the minimum angle between any two vectors in the spans of the two subsets (i.e. the first principal angle between the two spans). We can generalize the definition of coherence by taking the minimum angle between any pair of \emph{disjoint} subsets of at most $k$ columns. (If the two subsets share a column, then the angle is zero.)

\bdefn[Generalized coherence of degree $k$]
The generalized coherence of degree $k$ of some dictionary $\mv A$ , denoted by $\mu_k(\mv A)$, is the cosine of the minimum first-principal-angle between any two subspaces spanned by two disjoint subsets of $\leq k$ columns of $\mv A$:
\begin{equation}
\label{eq:defn-u_k}
\mu_k(\mv A)\quad:=\quad
\max_{I,J\subseteq [N] \;|\; (|I|,|J|\leq k) \;\wedge\; (I \cap J=\emptyset)}\quad
\max_{\mv x, \mv y \;|\; \|\mv{A}_I \mv x\|=\|\mv{A}_J \mv y\|=1}\quad
\left|\langle\mv{A}_I \mv x, \mv{A}_J \mv y\rangle\right|
\end{equation}
\edefn
Next, we extend Theorem~\ref{thm:list-decoding-spherical} to the case of $k>1$:
\bthm[Bound on $L(\mv A, k,\eps,1)$ in terms of generalized coherence $\mu_k(\mv A)$]
\label{thm:list-decoding-k>1}
Given a dictionary $\mv A$, a sparsity bound $k$, and an error bound $\epsilon$;  if $\eps<\sqrt{1-\mu_k(\mv A)}$, then the list size $L(\mv A, k, \eps, 1)$ 
 is bounded by
\begin{equation}
\label{eq:list-decoding-k>1}
L(\mv A, k, \eps, 1)\leq \left\lfloor\frac{1}{1-\frac{\epsilon^2}{1-\mu_k(\mv A)}}\right\rfloor.
\end{equation}
\ethm

Next, we present some properties of $\mu_k$.
\bprop[Generalized coherence as a function of $k$]
\label{prop:mu-k-gen}
For a fixed dictionary $\mv A$, the generalized coherence $\mu_k(\mv A)$ is non-decreasing with $k$. It reaches 1, exactly when $k=\lceil\spark{\mv{A}}/2\rceil$:
$$0 \leq \mu_1(\mv A) \leq \mu_2(\mv A) \leq \mu_3(\mv A) \leq \ldots \leq 1.$$
$$k\geq\left\lceil\spark{\mv{A}}/2\right\rceil \quad\Leftrightarrow\quad \mu_k(\mv A)=1.$$
\eprop

\bprop[Upperbound on generalized coherence]
\label{prop:mu_k-upperbound}
Given a dictionary $\mv A$ and an integer $k>1$ such that $\mu(\mv A)<\frac{1}{k-1}$, the generalized coherence $\mu_k(\mv A)$ is bounded by
$\mu_k(\mv A) \le \frac{k\cdot\mu(\mv A)}{1-(k-1)\cdot\mu(\mv A)}$.
\eprop

\bcor[Another upperbound on generalized coherence]
\label{prop:mu_k-upperbound2}
Given a dictionary $\mv A$ and an integer $k>1$, the generalized coherence $\mu_k(\mv A)$ is bounded by
$\mu_k(\mv A) \le (2k-1)\cdot\mu(\mv A)$.
\ecor

\bcor[Bound on $L(\mv A, k,\eps,1)$ in terms of traditional coherence $\mu(\mv A)$]
\label{cor:list-decoding-k>1}
Given a dictionary $\mv A$, a sparsity bound $k$, and an error bound $\epsilon$;  if
$\mu(\mv A)<\frac{1}{2k-1}$ and $\eps<\sqrt{\frac{1-(2k-1)\cdot\mu(\mv A)}{1-(k-1)\cdot\mu(\mv A)}}$, then the list size $L(\mv A, k, \eps, 1)$ is bounded by
\begin{equation}
L(\mv A, k, \eps, 1)\leq \left\lfloor\frac{1}{1-\frac{\left[1-(k-1)\cdot\mu(\mv A)\right]\epsilon^2}{1-(2k-1)\cdot\mu(\mv A)}}\right\rfloor.
\end{equation}
\ecor

In Appendix~\ref{app:kerdock}, we give an example showing that the condition $\mu(\mv A)<\frac{1}{2k-1}$ is necessary (up to constants).

\bprop[Tightness of generalized coherence upperbound]
\label{prop:mu_k-lowerbound}
For every integer $k>1$ and real number $0\leq u < \frac{1}{k-1}$, there exists a dictionary $\mv A$ whose coherence $\mu(\mv A)=u$ and whose generalized coherence satisfies
$\mu_k(\mv A) = \min\left(\frac{k\cdot\mu(\mv A)}{1-(k-1)\cdot\mu(\mv A)}, 1\right)$.
In particular, Proposition~\ref{prop:mu_k-upperbound} is tight.
\eprop

\subsubsection{Relationship to known results}
\label{sec:recover-old-results}
Results in this section recover and extend some well-known results in sparse representation. The following two standard results provide conditions for exact representation uniqueness. The first one depends on $\spark{\mv A}$ while the second depends on the coherence $\mu(\mv A)$. They are both subsumed by our results.
\bi
\item Theorem~\ref{thm:uniqueness-via-spark}: This result appeared in \cite{DonohoElad:2003} and as Theorem 2 in \cite{BrucksteinDonohoElad:Review2009}. It can be inferred from Theorem~\ref{thm:list-decoding-k>1} along with Proposition~\ref{prop:mu-k-gen} as follows: Exact representation means $\epsilon=0$. From Proposition~\ref{prop:mu-k-gen}, $k<1/2\spark{\mv A}$ if and only if $\mu_k(\mv A)<1$. From Theorem~\ref{thm:list-decoding-k>1}, if $\mu_k(\mv A)<1$, then $L(\mv A, k,\eps,1)\leq 1$, which means that the representation is unique.

\item Theorem~\ref{thm:uniqueness-via-mu}: This result follows from \cite{Donoho:1989:UPS:64936.64952} and appears as Theorem 5 in \cite{BrucksteinDonohoElad:Review2009}. It can be inferred from Corollary~\ref{cor:list-decoding-k>1}. Exact representation means $\epsilon=0$. By rearrangement, $k<1/2(\mu(\mv A)^{-1}+1)$ if and only if $\mu(\mv A)<1/(2k-1)$. From Corollary~\ref{cor:list-decoding-k>1}, if $\mu(\mv A)<1/(2k-1)$, then $L(\mv A, k, 0, 1)\leq 1$, which means that the representation is unique for sparsity $k$.
\ei

\section{\listsparse}
\label{sec:listsparse}

The essence of \sparse is to solve ${N\choose \leq k}$ least-squares problems. 
For every set $S$ with at most $k$ columns of $\mv A$, we have to solve
$\min_{\mv c}\| \mv A_{S} \mv c-\mv b \|_2$.
(Here, $\mv A_S$ denotes the matrix $\mv A$ restricted to columns in $S$.)
Notice that if the columns of $\mv A$ in $S$ are linearly independent, then 
the above problem has a unique solution. Otherwise, it has an infinite number 
of solutions.
The list contains all equally-good solutions from all ${N\choose \leq k}$ 
least-squares problems (i.e., solutions that produce the minimum error).

\begin{prop}[Necessary and sufficient condition for the list size to be finite] 
\label{prop:finite_list}
	Given an $m\times N$ matrix $\mv A$ and an integer $k \in [N]$;
    then, $s(\mv A, \mv b, k) < \infty$ for all $\mv b \neq \mv 0$ if and only if
	\[
	     k < \spark{\mv A}.
	\]
\end{prop}

The above proposition focused on one rather coarse condition on the size of the 
list, namely its finiteness. We next seek a more refined accounting. Given 
$k$ and $L$, where $L$ is a positive integer, we would like to find necessary 
and/or sufficient conditions on $\mv A$ so that
$ 
 	s(\mv A, \mv b, k) \leq L
$
for all $\mv b \neq \mv 0$. This question might be too hard. A relaxed version is 
to find asymptotic conditions on the ratio $m/N$ so that 
$s(\mv A, \mv b, k) \leq L$ for all $\mv b \neq \mv 0$. We conjecture that there 
is a trade-off between how small the ratio $m/N$ is and how small $L$ can be. 
We next present a few weaker results.
The first proposition is straightforward.

\begin{prop}[Sufficient condition for the list size to be $\leq L$]
\label{prop:suf}
Given an $m\times N$ matrix $\mv A$ and positive integers $k$ and $L$.
Then, $s(\mv A, \mv b, k) \leq L$ for all $\mv b \neq \mv 0$ provided that
\begin{equation}\label{eq:suf}
  k<\spark{\mv A} \wand {N \choose k}\leq L.
\end{equation}
\end{prop}

The next lemma provides a simple lowerbound for the maximum list size
over all non-zero $\mv b$, i.e. the quantity 
$\max_{\mv b \neq 0} s(\mv A, \mv b, k)$. 
This lowerbound holds for all but one trivial value of $k$ (i.e. $k=N$).
 
\begin{lmm}
\label{lem:bad_b}
If $k < N$, then there is a vector $\mv b \neq 0$ such 
that $s(\mv A, \mv b, k) > k$. 
\end{lmm}

We are now ready to derive a simple necessary condition for the list size
$s(\mv A, \mv b, k)$ to be bounded by $L$ {\em for all} $\mv{b \neq 0}$;
the necessary condition is that either the matrix $\mv A$ is a (square)
non-singular matrix (the trivial case), or the sparsity $k$ has to be
smaller than both $L$ and $\spark{\mv A}$.

\begin{prop}[Necessary condition for the list size to be $\leq L$]
\label{prop:nec}
Given an $m\times N$ matrix $\mv A$ and positive integers $k\in [N]$ and $L$; if 
$s(\mv A, \mv b, k) \leq L$ for all $\mv b \neq \mv0$ then
 \begin{equation}\label{eq:nec}
     \text{ either } k=N=\rank{\mv A} \text{ or } k < \min\{L, \spark{\mv A}\}.
\end{equation}
\end{prop}

The above condition is necessary but not sufficient for the list size to
be bounded by $L$ for all $\mv{b \neq 0}$. We are only able to derive necessary
and sufficient conditions when $L \in \{1,2\}$.

\begin{prop}[Necessary and sufficient condition for the list size to be $1$]
\label{prop:list-1}
Given an $m\times N$ matrix $\mv A$ and an integer $k \in [N]$; 
then $s(\mv A, \mv b, k)=1$ for all $\mv b \neq \mv 0$ if and only if
\begin{equation}\label{eq:unique-sol}
   k = N =\rank{\mv A}.
\end{equation}
(In particular, since $m\leq N$ it is implicit that $m=N$ too.)
\end{prop}

\begin{prop}[Necessary and sufficient condition for the list size to be $\leq 2$]
\label{prop:list-2}
Given an $m\times N$ matrix $\mv A$ and an integer $k \in [N]$;
then $s(\mv A, \mv b, k)\leq 2$ for all $\mv b \neq \mv 0$ if and only if
	 \begin{equation}\label{eq:l<=2}
	 \left[k<\spark{\mv A} \wand {N \choose k}\leq 2\right]\wor
	 \left[k=1 \wand  \rank{\mv A}=2 \wand \spark{\mv A}=3\right]
	 \end{equation}
\end{prop}

Next, we show that Lemma~\ref{lem:bad_b} is tight.
 
\begin{prop}[Minimum over all dictionaries of the worst-case list size is $k+1$]
\label{prop:tight-k+1}
Given $1 \leq k < m \leq N$,
\[
\min_{\mv A \in \R^{m \times N}} \max_{\mv b\neq 0} s(\mv A, \mv b, k) = k + 1.
\]
\end{prop}

The previous propositions focused on the worst-case performance for all
measurement vectors $\mv{b}$. It is natural to answer similar questions
regarding the average case. 
We do so for list size being finite and one. The bounds are better in the random case similar to what is known in the list decoding setting (see e.g.~\cite{RU10}).

\begin{prop}[Necessary and sufficient condition for the list size to be finite with probability of 1 assuming that $\mv b$ is chosen uniformly from the unit sphere]
\label{prop:random-b-rank}

	Given an $m\times N$ matrix $\mv A$ and an integer $1 \leq k \leq N$,
    \[ 
		\pr_{\mv b: \|\mv b\|=1} \Big( s(\mv A, \mv b, k) < \infty \Big) =1
	\]
	if and only if 
	\[
		k\leq \rank{\mv A}.
	\]
\end{prop}

\begin{prop}[Necessary and sufficient condition for the list size to be $1$ with probability of 1 assuming that $\mv b$ is chosen uniformly from the unit sphere]
\label{prop:random-b}
	 Given an $m\times N$ matrix $\mv A$ and an integer $1 \leq k \leq N$,
	 \[ 
	  	\pr_{\mv b: \|\mv b\|=1} \Big( s(\mv A, \mv b, k) = 1\Big) = 1
	\]
	if and only if 
	\[
		k<\spark{\mv A}-1.
	\]
\end{prop}

\section{Examples}
\label{sec:examples}

We present a simple family of examples that shows that the dependence of $\eps$ on $\mu$ and $k$ in Proposition~\ref{prop:av-p2-k=1} and Proposition~\ref{prop:av-p2-k} (with $\gamma=0$) are in the right ballpark. In particular, we will show that
\begin{lmm}
\label{lem:tight-eg}
For  every $k\ge 1$ and $\eps>0$ and large enough $m$, there exists a matrix $\mv A^*$ with $m$ rows such that
$\mu= \frac{1}{1+\eps^2\cdot k}$,
and
$s(\mv A^*, k,\eps,0)\ge \frac{m-1}{k}$.
\end{lmm}

We present some implications. Note that for $k=1$, we have $\mu \ge 1-\eps^2$, which shows that the bound of $\eps\le \sqrt{1-\mu}$ in Proposition~\ref{prop:av-p2-k=1} is necessary. For $k>1$, (and say $\eps=1/\sqrt{2}$)
the above implies that if $\mu >\frac{1}{2k}$, then we cannot hope to have a list size bounded only in terms of $\eps$ and $k$. 
This implies that the bound of $\mu \le \frac{2}{k}$ in Proposition~\ref{prop:av-p2-k} (with $\gamma=0$) is necessary (up to a factor of $4$).

In Appendix~\ref{app:kerdock} we show that the various bounds on $\mu$ in our upper bounds are necessary. This involves using the Kerdock code and it subsumes the spike and sines example.  For pedagogical reasons,
we present the spikes and sines example in Appendix~\ref{app:spikes-and-sines}. 

\section{Open Questions}
\label{sec:concl}

We conclude with two major open questions. In this work, we presented bounds on the list size. However, these results are purely combinatorial. It would be extremely useful to present algorithms that can solve the \listapprox and \listsparse problems. To begin with presenting a polynomial time algorithms when the list sizes are bounded by $N^{O(1)}$ for any non-trivial matrix $\mv A$ would be very interesting.

The other tantalizing question left open by our work is to present applications of the new notions of \listapprox and \listsparse to parallel the well documented applications of list decoding in complexity theory.

{

  \bibliographystyle{acm}
  \def\DIR#1{\gdef\@DIR{#1}}

  \DIR{.}
  \bibliography{list-sparse}
}

\appendix

\section{Image Compression Example From Section~\ref{sec:intro}} \label{sec:image_appendix}

The original image, shown in the upper left of Figure~\ref{fig:blobs}, is $m = 256 \times 256$ pixels and we compute four levels of the two-dimensional Haar wavelet packet decomposition which gives us four times as many coefficients as the original number of pixels. The coherence of this redundant dictionary is $1/\sqrt{2}$. The wavelet packet decomposition is arranged in a quad-tree (in our example, a depth four quad-tree) and the total number of coefficients in all of the nodes at a fixed level is $m$. The coefficients in any maximal anti-chain in the tree form an orthonormal basis and we use this feature to construct our three different sparse representations. 
\begin{enumerate}
	\item \textbf{Class 1: Large and random medium coefficients of a fixed basis.} We consider all the wavelet packet coefficients in a particular orthonormal basis, the one formed by the terminal nodes of the depth four quad tree. We threshold the coefficients and retain only the large coefficients which are greater than $10^{-1}$ in absolute value. Then, we include, uniformly at random, half of the medium coefficients which are between $10^{-2}$ and $10^{-1}$ in absolute value. The total sparsity is $0.20$ of the total pixels. 
	\item \textbf{Class 2: Truncated Entropy BestBasis.} We compute the ``Best Basis'' according to the Shannon entropy function (see~\cite{CoifmanWickerhauserBB:1992} for details on the Best Basis algorithm for wavelet packets), truncate the coefficients in this orthonormal basis, retaining the largest 20\% in absolute value. As the basis is selected specifically for its compressive capabilities, there are only $0.11$ fraction of coefficients non-zero for this representation\footnote{We note that this is one of the original heuristic uses of the Best Basis algorithm: compute the Best Basis and then truncate the coefficients.}. 
	\item \textbf{Class 3: Truncated $\ell_1$ BestBasis.} We perform the same computation as in the Class 2 construction but we use the $\ell_1$ norm of the wavelet packet coefficients as the ``entropy'' function; i.e., we choose the orthonormal basis that has minimal $\ell_1$ norm. Then we truncate these coefficients, keeping the top 20\% in absolute value, and, as in the previous construction, only 10\% of the coefficients are actually non-zero. 
\end{enumerate}
In Figure~\ref{fig:blobs} upper right, lower left and right, we show the three different list sparse approximations of the original image. The representations from Class 1 are not as accurate as those from Class 2 and 3. We can see some artifacts in the reconstruction in the upper right of Figure~\ref{fig:blobs}.

Figure~\ref{fig:rep_stats} shows the statistics of the six different list sparse approximations we construct. Representation numbers 1-4 are instances of Class 1 and the last two (representation numbers 5 and 6) are the two different BestBasis constructions, Class 2 and 3, respectively. All of the representations use no more than 20\% of the original coefficients and are each within a relative error of $0.01$ of the original image. In fact, the BestBasis constructions are exact (up to numerical precision) but their sparsity values differ which show that the two classes choose different orthonormal bases and are, indeed, different constructions. They are not, however, all equally efficient to compute. The BestBasis algorithm, which we use for Class 2 and 3, for an image of size $m$ requires time $O(m \log m)$ to compute (on top of the original wavelet packet decomposition) while the first type of sparse approximation, Class 1, does not require any additional computation (beyond thresholding the terminal node coefficients).

\begin{figure}[ht!]
	\begin{center}
		\includegraphics[height=4in,width=5.25in]{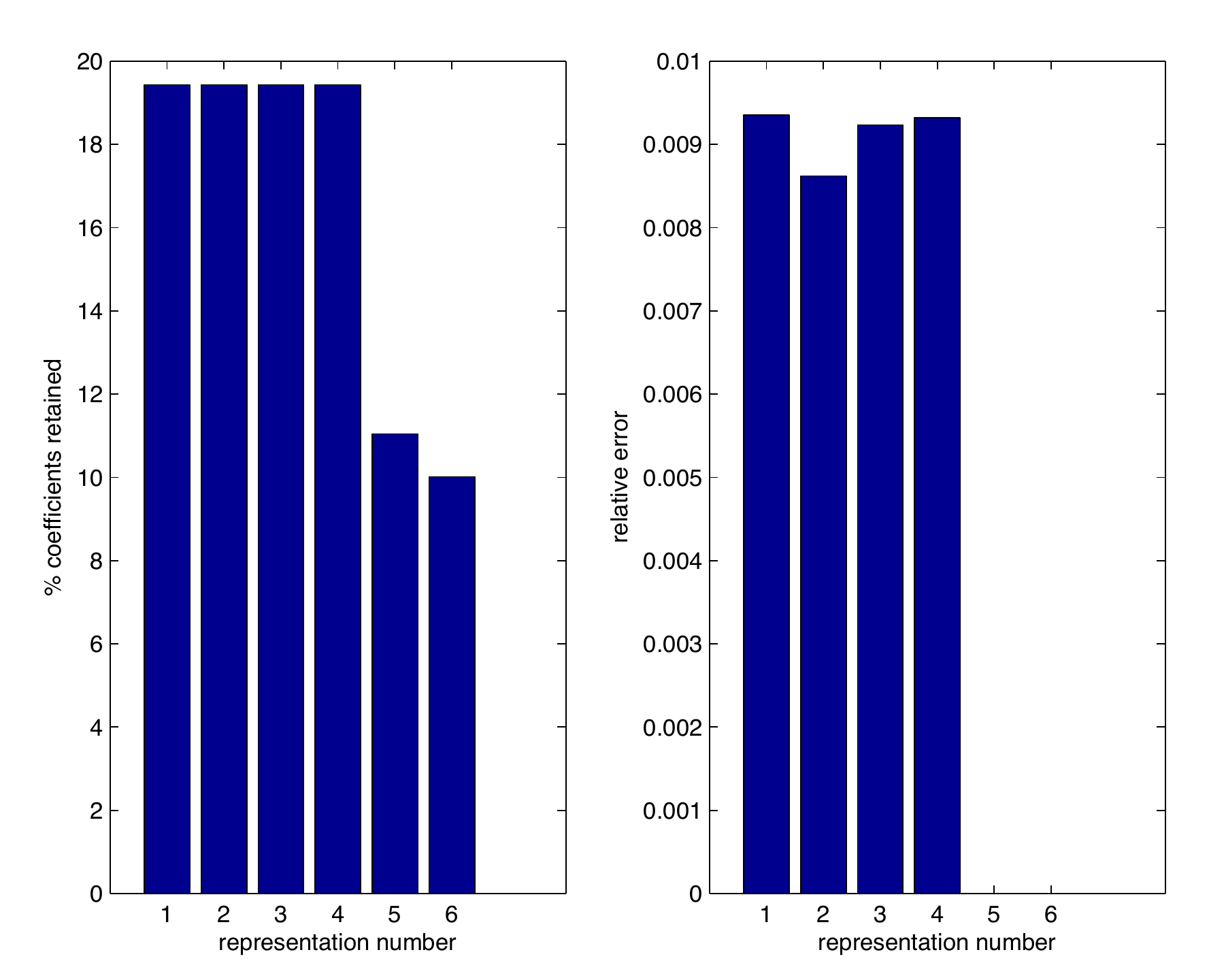}
	\end{center}
	\caption{(Left.) We plot the sparsity of the compressed representations as a function of the representation type. The first four are of the first type and the last two are the BestBasis constructions. (Right.) We plot the relative error of the different representations as a function of the representation types.}
	\label{fig:rep_stats}
\end{figure}

\section{Missing Proofs from Section~\ref{sec:listapprox}}

\subsection{Proof of Proposition~\ref{prop:mu=0-exp-lb}}
\begin{proof}
Let $\mv{A}$ be the $m\times m$ identity matrix. Note that this matrix has $\mu=0$ and that $N=m$.

Now consider the input vector $\vb = (b_1,...,b_m)$ (in standard basis), where

\[b_1^2 = 1- (m-1)\cdot \frac{\eps^2}{m-k},\]

and for every $i>1$,

\[b_i^2 = \frac{\eps^2}{m-k}.\]

Note that in this case for every $\Lambda\subseteq [m]$ with $|\Lambda|=k$ and $1\in \Lambda$, $\min_{\mv{x}\in\R^k}\|\mv{A}_{\Lambda}\mv{x}-\vb\|_2 \le \eps$. Further, every strict subset of $\Lambda$ has error $>\eps$. Note that with the above choice of parameters, $b_1^2 > \eps^2/(N-k)$, so it makes sure that every optimal $\Lambda$ has to have $1$ in it.
\end{proof}

\subsection{List size bound $L(\mv{A},1,\eps,o(L)$)}

\label{sec:list-approx-k-1}

For $k= 1$, each sparse approximation can consist of a single atom in $\mv{A}$ only.
We now will argue the following result:

\begin{prop}
\label{prop:av-p2-k=1}
Assume that the coherence of the dictionary $\mv{A}$ is $\mu$. Then as long as
\[\eps \le \epsbound,\]
we have
\[L(\mv{A},1,\eps,1)\le \frac{4}{1-\eps^2}.\]
\end{prop}


In the rest of this subsection, we will prove Proposition~\ref{prop:av-p2-k=1}.
We begin by making the following simple observation. 


\begin{prop}
\label{prop:av-ver}
If the following is true for every $\Lambda\subseteq [N]$ with $|\Lambda|=L$ and every $\vb\in \R^m$ with $\|\vb\|_2=1$:
\begin{equation}
\label{eq:suff-cond}
\max_{\vb\in\R^m: \|\vb\|_2=1} \sum_{i\in \Lambda} \inabs{\inner{\vA_i,\vb}} < L\sqrt{1-\eps^2},
\end{equation}
then
\[L(\vA,1,\eps,1)\le L-1.\]
\end{prop}
\begin{proof}
Note that $L(\vA,1,\eps,1)\le L-1$ if and only if for every $\Lambda\subseteq [N]$ with $|\Lambda|=L$ and every $\vb\in \R^m$ with $\|\vb\|_2=1$, we have:
\[\max_{i\in\Lambda} \min_{x\in\R} \|\vA_i x-\mv{b}\|_2^2 >\eps^2.\]
Note that the $x$ that minimizes the inner quantity is given by $\inner{\vA_i,\vb}$. Further since $\|\inner{\vA_i,\vb}\cdot\vA_i-\vb\|_2^2=1-\inner{\vA_i,\vb}^2$, the above condition is satisfied if and only if
\[\max_{i\in\Lambda} \left(1-\inner{\vA_i,\vb}^2\right)>\eps^2,\]
which in turn is true if and only if
\[\min_{i\in\Lambda} \inabs{\inner{\vA_i,\vb}} <\sqrt{1-\eps^2}.\]
Since the average is bigger than the minimum, the condition in \eqref{eq:suff-cond} implies the above which completes the proof.
\end{proof}

Thus, to prove Proposition~\ref{prop:av-p2-k=1}, we need to show that with $\eps\le \epsbound$ and $L=\frac{4}{1-\eps^2}+1$, \eqref{eq:suff-cond} is satisfied, which is what we do next.

We begin with a simplification of \eqref{eq:suff-cond}.

\begin{clm}
\label{clm:max-support}
The following is true for any $\Lambda\subseteq [N]$
\[\max_{\vb\in\R^m: \|\vb\|_2=1} \sum_{i\in \Lambda} \inabs{\inner{\vA_i,\vb}} = \max_{\vb\in\mspn{\Lambda}: \|\vb\|_2=1} \sum_{i\in \Lambda} \inabs{\inner{\vA_i,\vb}},\]
where $\mspn{\Lambda}$ is shorthand for the span of the vectors $\{\vA_i\}_{i\in \Lambda}$.
\end{clm}
\begin{proof} The LHS is trivially larger than the RHS so we only need to show that the LHS is no bigger than the RHS. Towards this end, consider an arbitrary $\vb\in\R^m$ with $\|\vb\|_2=1$. Decompose $\vb=\vb_1+\vb_2$, where $\vb_1$ is the projection of $\vb$ onto $\mspn{\Lambda}$ and $\vb_2$ is the remainder. Note that $\vb_2$ is orthogonal to $\vA_i$ for every $i\in\Lambda$. Thus, we have
\[\sum_{i\in \Lambda} \inabs{\inner{\vA_i,\vb}} = \sum_{i\in \Lambda} \inabs{\inner{\vA_i,\vb_1}} \le \sum_{i\in \Lambda} \inabs{\inner{\vA_i,\frac{\vb_1}{\|\vb_1\|_2}}},\]
where the inequality follows from the fact that since $\vb_1$ is a projection onto $\mspn{\Lambda}$, $\|\vb_1\|_2\le 1$. The proof is complete by noting that
since $\frac{\vb_1}{\|\vb_1\|_2}$ is a unit vector and is in $\mspn{\Lambda}$, we have
\[\sum_{i\in \Lambda} \inabs{\inner{\vA_i,\frac{\vb_1}{\|\vb_1\|_2}}} \le \max_{\vb\in\mspn{\Lambda}: \|\vb\|_2=1} \sum_{i\in \Lambda} \inabs{\inner{\vA_i,\vb}}.\]
\end{proof}

Note that by Claim~\ref{clm:max-support} and the discussion above, to prove Proposition~\ref{prop:av-p2-k=1}, we need to show that $\eps\le \epsbound$ and $L=\frac{4}{1-\eps^2}+1$ implies
\begin{equation}
\label{eq:suff-cond-new}
\max_{\vb\in\mspn{\Lambda}: \|\vb\|_2=1} \sum_{i\in \Lambda} \inabs{\inner{\vA_i,\vb}} < L\sqrt{1-\eps^2},
\end{equation}
which we do next. Towards that end, we record the following result:

\begin{lmm}
\label{lem:beta-bounds}
Let $\vb=\sum_{j\in\Lambda} \beta_j\cdot\vA_i$. If 
\begin{equation}
\label{eq:mu-ub}
\mu<\frac{1}{L},
\end{equation}
then the following is true (where $\vbeta=(\beta_j)_{j\in\Lambda}$):
\[\|\vbeta\|_2^2\le \frac{1}{1-\mu L}.\]
\end{lmm}

We will prove the lemma shortly but for now we return to the proof of \eqref{eq:suff-cond-new}. Fix an arbitrary $\vb\in\mspn{\Lambda}$ with $\|\vb\|_2=1$. For now we claim that our choices of $\mu$ and $L$ imply that
\begin{equation}
\label{eq:mu-ub-2}
\mu \le \frac{1}{4(L-1)}.
\end{equation}
Let $\vb=\sum_{j\in\Lambda} \beta_j \vA_j$, where $\vbeta=(\beta_j)_{j\in\Lambda}$. Consider the following relationships:
\begin{align}
\sum_{i\in\Lambda} \inabs{\inner{\vA_i,\vb}} &= \sum_{i\in\Lambda}\inabs{\sum_{j\in\Lambda} \beta_j\inner{\vA_i,\vA_j}} \notag\\
\label{eq:suff-cond-s1}
&\le \sum_{i\in\Lambda}\sum_{j\in\Lambda} \inabs{\beta_j}\cdot\inabs{\inner{\vA_i,\vA_j}}\\
\label{eq:suff-cond-s2}
&\le \sum_{i\in\Lambda}\left(\inabs{\beta_i}+ \mu \sum_{j\in\Lambda\setminus\{i\}} \inabs{\beta_j}\right)\\
&=\|\vbeta\|_1 (1+\mu(L-1)) \notag\\
\label{eq:suff-cond-s3}
&\le \sqrt{L}\cdot \|\vbeta\|_2\cdot (1+\mu L)\\
\label{eq:suff-cond-s4}
&\le \sqrt{L}\cdot \frac{1}{\sqrt{1-\mu L}}\cdot (1+\mu L)\\
\label{eq:suff-cond-s5}
&< 2\sqrt{L}.
\end{align}
In the above, \eqref{eq:suff-cond-s1} follows from the triangle inequality, \eqref{eq:suff-cond-s2} follows from the fact that $\vA$ has coherence $\mu$, \eqref{eq:suff-cond-s3} follows from Cauchy-Schwarz, \eqref{eq:suff-cond-s4} follows from Lemma~\ref{lem:beta-bounds} (and the fact that \eqref{eq:mu-ub-2} implies \eqref{eq:mu-ub} assuming $L>4/3$, which holds for our choice of $L$), and \eqref{eq:suff-cond-s5} follows from \eqref{eq:mu-ub-2} (and the fact that $\mu\leq 1/17$, which will be true assuming that there is an $\epsilon$ such that $\eps\le \epsbound$).

Note that by \eqref{eq:suff-cond-s5}, the required relation \eqref{eq:suff-cond-new} is satisfied if
\[\sqrt{4L} < L\sqrt{1-\eps^2},\]
which is satisfied by our choice of $L=\frac{4}{1-\eps^2}+1$. Next, we argue that $\mu\le \frac{1}{4(L-1)}$, which would imply \eqref{eq:mu-ub-2}. Indeed, we picked $\eps \le \epsbound$, which in turn implies that
\[\mu \le \frac{\left(1-\eps^2\right)}{17}\le \frac{1}{4}\cdot\left(\frac{1-\eps^2}{17/4}\right)\le \frac{1}{4(L-1)},\]
as desired. The proof of Proposition~\ref{prop:av-p2-k=1} is complete except for the proof of Lemma~\ref{lem:beta-bounds}, which we present next.

\begin{proof}[Proof of Lemma~\ref{lem:beta-bounds}]

Consider the following sequence of relations:
{\allowdisplaybreaks
\begin{align}
\|\vb\|_2^2 &= \sum_{j=1}^L \inabs{\beta_j}^2\cdot \|\vA_1\|_2^2 +\sum_{i\neq j\in [L]} \beta_i\beta_j \inner{\vA_i,\vA_j} \notag\\
& \ge \sum_{j=1}^L \inabs{\beta_j}^2\cdot \|\vA_1\|_2^2 -\sum_{i\neq j\in [L]} \inabs{\beta_i}\inabs{\beta_j} \inabs{\inner{\vA_i,\vA_j}}\notag\\
& = \|\vbeta\|_2^2 -\sum_{i\neq j\in [L]} \inabs{\beta_i}\inabs{\beta_j} \inabs{\inner{\vA_i,\vA_j}}\notag\\
\label{eq:beta-norm-s1}
& \ge \|\vbeta\|_2^2 -\mu\sum_{i, j\in [L]} \inabs{\beta_i}\inabs{\beta_j} \\
& = \|\vbeta\|_2^2 -\mu\left(\sum_{i\in [L]} \inabs{\beta_i}\right)^2 \notag\\
\label{eq:beta-norm-s2}
&\ge \|\vbeta\|_2^2 - \mu L\|\vbeta\|_2^2.
\end{align}
}
In the above \eqref{eq:beta-norm-s1} follows from the fact that $\vA$ has coherence $\mu$. \eqref{eq:beta-norm-s2} follows from Cauchy-Schwarz.
\eqref{eq:beta-norm-s2} along with the fact that $\|\vb\|_2=1$ proves the bound on $\|\vbeta\|_2^2$. 

\end{proof}

\subsection{Proof of Proposition~\ref{prop:av-p2-k}}

The proof of Proposition~\ref{prop:av-p2-k} follows the structure of the proof of Proposition~\ref{prop:av-p2-k=1} so we will omit details for arguments that are similar to the ones we made earlier.

We begin with our choice of $L$:

\begin{equation}
\label{eq:k-L-def}
L= \left\lceil \left(\frac{11}{1-\eps^2}\right)^{1/(1-\gamma)}\right\rceil.
\end{equation}
We claim that the above along with \eqref{eq:k-mu-lb-org} implies that:
\begin{equation}
\label{eq:k-mu-lb}
\mu \le \frac{1}{2kL}.
\end{equation}
Indeed \eqref{eq:k-mu-lb-org} implies that
\[1-\eps^2 \ge 12(2k\mu)^{1-\gamma},\]
which along with the definition of $L$ in \eqref{eq:k-L-def} implies \eqref{eq:k-mu-lb}.

Let $S=\{\Lambda_1,\dots,\Lambda_L\}$ be an arbitrary collection of $L$ subsets of $[N]$ of size exactly $k$ such that every column in $[N]$ occurs in at most $L^{\gamma}$ of the sets in $S$. To prove the claimed result, it suffices to show that for every $\mv{b}\in\R^m$, we have
\[\max_{\Lambda\in S} \min_{\mv{x}\in \R^k} \|\mv{A}_{\Lambda}\mv{x}-\mv{b}\|_2^2 >\eps^2.\]
The above condition is the same as showing
\[\min_{\Lambda\in S} \|\mv{b}_{\Lambda}\|_2^2 < 1-\eps^2,\]
where $\mv{b}_{\Lambda}$ is the projection of $\mv{b}$ to $\mspn{\Lambda}$. A sufficient condition for the above to be satisfied is
\begin{equation}
\label{eq:k-suff}
\sum_{\Lambda\in S} \|\mv{b}_{\Lambda}\|_2^2 < L(1-\eps^2).
\end{equation}
For the rest of the proof, we will show that if \eqref{eq:k-mu-lb} and \eqref{eq:k-L-def} are true then the above condition is satisfied.

For notational convenience, define
\[\Lambda'=\cup_{\Lambda\in S} \Lambda.\]
Given this definition, we can assume WLOG that $\mv{b}$ is in $\mspn{\Lambda'}$. (This is the generalization of Claim~\ref{clm:max-support} to general $k$.) In other words, WLOG assume that 
\begin{equation}
\label{eq:b-decompose}
\vb=\sum_{i\in\Lambda'} \beta_i\cdot \mv{A}_i.
\end{equation}
Since \eqref{eq:k-mu-lb} satisfies the condition on $\mu$ in Lemma~\ref{lem:beta-bounds}, we have that
\[\|\vbeta\|_2^2\le \frac{1}{1-\mu kL},\]
since $|\Lambda'|\le kL$.

In our notation, \cite[Lemma 3]{approx-parseval} implies that
\begin{equation}
\label{eq:approx-parseval}
 \|\vb_{\Lambda}\|_2^2 \le \frac{1}{1-\mu k}\cdot \sum_{\ell\in\Lambda} \langle \mv A_{\ell},\vb\rangle^2.
\end{equation}

Now consider the following sequence of relations:

{\allowdisplaybreaks
\begin{align}
\label{eq:bl-step2}
\|\vb_{\Lambda}\|_2^2&\le \frac{1}{1-\mu k}\cdot \sum_{\ell\in\Lambda} \langle \mv A_{\ell},\sum_{j\in {\Lambda'}} \beta_j \cdot \mv A_j\rangle^2 \\
&= \frac{1}{1-\mu k} \cdot \left( \sum_{\ell\in\Lambda}\left(\beta_{\ell}+ \sum_{j\in\Lambda'\setminus\{\ell\}} \beta_j\langle \mv{A}_{\ell},\mv{A}_j\rangle\right)^2\right) \notag\\
\label{eq:bl-step4}
&\le2\cdot \frac{1}{1-\mu k} \cdot \left( \sum_{\ell\in\Lambda}\left(\left(\beta_{\ell}\right)^2+ \left(\sum_{j\in\Lambda'\setminus\{\ell\}} \beta_j\langle \mv{A}_{\ell},\mv{A}_j\rangle\right)^2\right)\right)\\
\label{eq:bl-step5}
&\le2\cdot \frac{1}{1-\mu k} \cdot \left( \sum_{\ell\in\Lambda}\left(\left(\beta_{\ell}\right)^2+ \left(\sum_{j\in\Lambda'\setminus\{\ell\}} |\beta_j|\cdot |\langle \mv{A}_{\ell},\mv{A}_j\rangle|\right)^2\right)\right)\\
\label{eq:bl-step6}
&\le 2\cdot \frac{1}{1-\mu k} \cdot \left( \sum_{\ell\in\Lambda}\left(\left(\beta_{\ell}\right)^2+ \mu^2\cdot\left(\sum_{j\in\Lambda'\setminus\{\ell\}} |\beta_j|\right)^2\right)\right)\\
\label{eq:bl-step7}
&\le 2\cdot \frac{1}{1-\mu k} \cdot \left( \sum_{\ell\in\Lambda}\left(\left(\beta_{\ell}\right)^2+ \mu^2kL\cdot\sum_{j\in\Lambda'} \left(|\beta_j|\right)^2\right)\right)\\
&=2\cdot \frac{1}{1-\mu k} \cdot \left( \sum_{\ell\in\Lambda}\left(\left(\beta_{\ell}\right)^2+ \mu^2kL\cdot\|\vbeta\|_2^2\right)\right) \notag\\
\label{eq:bl-step9}
&\le2\cdot \frac{1}{1-\mu k} \cdot \left( \sum_{\ell\in\Lambda}\left(\left(\beta_{\ell}\right)^2+ \mu^2kL\cdot\frac{1}{1-\mu kL}\right)\right)\\
\label{eq:bl-step10}
&\le 4 \cdot \left( \sum_{\ell\in\Lambda}\left(\beta_{\ell}\right)^2+ 2\mu^2k^2L\right)
\end{align}
}

In the above, \eqref{eq:bl-step2} follows from \eqref{eq:b-decompose} and \eqref{eq:approx-parseval}. \eqref{eq:bl-step4} follows from Cauchy-Schwartz. \eqref{eq:bl-step5} follows since we replaced some potentially negative terms with positive terms, while \eqref{eq:bl-step6} follows from the definition of $\mu$. \eqref{eq:bl-step7} follows from Cauchy-Schwartz. \eqref{eq:bl-step9} follows from Lemma~\ref{lem:beta-bounds} and the bound on $\mu$ from \eqref{eq:k-mu-lb}. Finally, \eqref{eq:k-mu-lb} implies \eqref{eq:bl-step10}.

\eqref{eq:bl-step10} implies that
{\allowdisplaybreaks
\begin{align}
\sum_{\Lambda\in S}\|\vb_{\Lambda}\|_2^2 &\le 4\sum_{\Lambda\in S}\sum_{\ell\in\Lambda} \beta_{\ell}^2 +8\mu^2 k^2 L^2 \notag\\
\label{eq:k-final-s1}
 &\le 4L^{\gamma}\cdot \sum_{\ell\in\Lambda'} \beta_{\ell}^2 +8\mu^2 k^2 L^2\\
\label{eq:k-final-s2}
 &\le 4L^{\gamma}\cdot \frac{1}{1-\mu kL} +8\mu^2 k^2 L^2\\
\label{eq:k-final-s3}
 &\le 8L^{\gamma} +2 \\
\label{eq:k-final-s4}
 &\le 8L^{\gamma} +\frac{1}{4}\cdot L(1-\eps^2)\\
\label{eq:k-final-s5}
&< L(1-\eps^2).
\end{align}
In the above, \eqref{eq:k-final-s1}  follows from the fact that each atom appears in at most $L^{\gamma}$ sets in $S$. \eqref{eq:k-final-s2} follows from Lemma~\ref{lem:beta-bounds}, while \eqref{eq:k-final-s3} follows from \eqref{eq:k-mu-lb}. \eqref{eq:k-final-s4} follows from the subsequent argument. \eqref{eq:k-L-def} implies that
\[L\ge \frac{8}{1-\eps^2},\]
which in turn implies \eqref{eq:k-final-s4}.
\eqref{eq:k-final-s5} follows from \eqref{eq:k-L-def}.

\eqref{eq:k-final-s5} completes the proof.
}

\subsection{Euclidean codes}
\label{app:euclid}

\bdefn[Euclidean code]
In a Euclidean code, the codewords are points in a Euclidean space. The distance $\delta$ of the code is the minimum Euclidean distance between any pair of points.
\edefn

Given a Euclidean code of distance $\delta$, a received word $\mv w$, and an error bound $\epsilon$, the list decoding problem is to output a list of all codewords that are within Euclidean distance $\epsilon$ from $\mv w$.

\begin{lmm}[Radius of a sphere circumscribed about a regular simplex]
Given a regular simplex of $n$ vertices and unit edges, the radius of the circumscribed sphere is
\begin{equation}
\label{eq:simplex-radius}
\sqrt{\frac{n-1}{2n}}
\end{equation}
\end{lmm}
\bp
Consider the regular simplex whose $n$ vertices are $(1, 0, \ldots, 0)$, $(0, 1, 0, \ldots, 0)$, \ldots, $(0, \ldots, 0, 1)$. The center of the circumscribed sphere is $(1/n, \ldots, 1/n)$. The distance between this center and any vertex (i.e., the radius) is
\[\sqrt{\left(\frac{n-1}{n}\right)^2+(n-1)\left(\frac{1}{n}\right)^2}=\sqrt{\frac{n-1}{n}}\]
The distance between any two vertices (i.e. the edge length) is $\sqrt{2}$. \eqref{eq:simplex-radius} follows.
\ep

\bthm[Bound on list-decoding of Euclidean codes]
\label{thm:list-decoding-euclidean}
Given a Euclidean code of distance $\delta$ and an error bound $\epsilon$; if $\eps<\delta/\sqrt{2}$, then the maximum list size $L$ for any received word is bounded by
\begin{equation}
\label{eq:list-decoding-euclidean}
L\leq \left\lfloor\frac{1}{1-\frac{2\epsilon^2}{\delta^2}}\right\rfloor.
\end{equation}
Moreover, the above bound is tight if the right-hand side is $\leq m+1$, where $m$ is the dimension of the code.
\ethm
\bp
Bounding the list size in Euclidean codes corresponds to packing spheres (corresponding to codewords) whose radii are $\delta/2$ inside one big sphere whose radius is $\eps+\delta/2$ (this sphere corresponds to the received word).
It is a well-known fact that the tightest packing of spheres \cite{conway1999sphere, bor1978sphere} is achieved when the centers of adjacent spheres form a simplex. 
Therefore for every integer $n>1$, if $\epsilon$ is less than the radius of a sphere that circumscribes a regular simplex of $n$ vertices whose edge length is $\delta$, then the list size $L$ is less than $n$. Substituting from \eqref{eq:simplex-radius}, we get
$$\epsilon<\sqrt{\frac{n-1}{2n}}\delta  \quad\Rightarrow\quad L<n$$
\footnote{The special case of $n=2$ is used very often in error-correcting codes: If $\eps<\delta/2$, then we have a unique decoding.}
By rearranging the left-hand side of the above implication under the assumption $\eps<\delta/\sqrt{2}$, we get
$$\frac{1}{1-\frac{2\epsilon^2}{\delta^2}}<n  \quad\Rightarrow\quad L<n$$
for all integers $n>1$. All those implications can be combined in \eqref{eq:list-decoding-euclidean}.
Moreover, a regular simplex of $n$ vertices exists in $m$ dimensions if $n \leq m+1$. The tightness condition follows.
\ep

\subsection{Proof of Theorem~\ref{thm:list-decoding-spherical}}

\bp
In \eqref{eq:spherical-list-decoding}, the point $\mv A_i x$ that minimizes $\| \mv A_i x -\mv b \|_2$ is simply the projection of $\mv b$ onto the line that extends the vector $\mv A_i$. Therefore, the value of $\min_{x}\| \mv A_i x -\mv b \|_2$ corresponds to the sine of the angle between this line and $\mv b$. The column $\mv A_i$ satisfies \eqref{eq:spherical-list-decoding} if and only if this sine is $\leq\epsilon$.

Columns of $\mv A$ are points on the unit sphere. The angle between any two vectors $\mv A_i$ and $\mv A_j$ is no less than $\arccos(\mu(\mv A))$. Therefore, the Euclidean distance between any two points $\mv A_i$ and $\mv A_j$ is no less than $\delta$, where $\delta$ is twice the sine of half the angle $\arccos(\mu(\mv A))$. Using the formula for sine half an angle, we get
\begin{equation}
\label{eq:spherical-delta}
\delta=2\sqrt{\frac{1-\mu(\mv A)}{2}}
\end{equation}

Based on the simplex argument that we used earlier in the proof of Theorem~\ref{thm:list-decoding-euclidean}, we have the following: For every integer $n>1$, if $\epsilon$ is less than the radius of a sphere that circumscribes a regular simplex of $n$ vertices whose edge length is $\delta$, then the list size $L$ is less than $n$. By substitution from \eqref{eq:simplex-radius} and \eqref{eq:spherical-delta}
$$\epsilon<\sqrt{\frac{n-1}{n}(1-\mu(\mv A))} \quad\Rightarrow\quad L(\mv A,1,\eps,1)<n$$
By rearranging the left-hand side of the above implication under the assumption $\eps<\sqrt{1-\mu(\mv A)}$, we get
$$\frac{1}{1-\frac{\epsilon^2}{1-\mu(\mv A)}}<n  \quad\Rightarrow\quad L(\mv A,1,\eps,1)<n$$
for all integers $n>1$. All those implications can be combined in \eqref{eq:list-decoding-spherical}.
Moreover, a regular simplex of $n$ vertices exists on the surface of a sphere in $m$ dimensions if $n \leq m$. The tightness condition follows.
\ep

\subsection{Proof of Theorem~\ref{thm:list-decoding-k>1}}
\bp
Suppose that we have a list of $L$ solutions such that no atom appears in more than one of them.
This means that we have $L$ disjoint subsets $\Lambda_1, \ldots, \Lambda_{L}$ of $\leq k$ columns of $\mv A$; each one of them gives a $k$-sparse representation $\mv A_{\Lambda_1}\mv x_1, \ldots, \mv A_{\Lambda_L}\mv x_L$ of $\mv b$ with an error $\leq \epsilon$. Therefore, the angle between $\mv b$ and any of the vectors $\mv A_{\Lambda_1}\mv x_1, \ldots, \mv A_{\Lambda_L}\mv x_L$ is no more than $\arcsin(\epsilon)$. By definition of generalized coherence $\mu_k(\mv A)$, the angle between any two of the vectors $\mv A_{\Lambda_1}\mv x_1, \ldots, \mv A_{\Lambda_L}\mv x_L$ is no less than $\arccos(\mu_k(\mv A))$. Using the same simplex argument that we used earlier in the proof of Theorem~\ref{thm:list-decoding-spherical}, we get \eqref{eq:list-decoding-k>1}.
\ep

\subsection{Proof of Proposition~\ref{prop:mu-k-gen}}

\bp
The first part is by definition. The second part holds because the generalized coherence is 1 if and only if there exists a single vector that is spanned by two disjoint subsets $I, J$ of $\leq k$ columns. This happens when columns in $I \cup J$ are linearly dependent, which happens when the spark of $\mv A$ is no more than $2k$.
\ep

\subsection{Proof of Proposition~\ref{prop:mu_k-upperbound}}

\bp
Fix a matrix $\mv A$ and for notational convenience define $\mu=\mu(\mv A)$. Fix arbitrary $I,J\subseteq [N]$ such that $|I|=|J|=k$ and the sets are disjoint. Further, let $\mv x$ and $\mv y$ be arbitrary vectors such that $\|\mv A_I \mv x\|_2=\|\mv A_J \mv y\|_2=1$. We will show that
\[\left|\langle \mv A_I \mv x, \mv A_J \mv y\rangle\right|\le \frac{\mu k}{1-\mu (k-1)},\]
which would complete the proof.

Indeed,
\begin{align*}
\left|\langle \mv A_I \mv x, \mv A_J \mv y\rangle\right|&\le \sum_{i\in I}\sum_{j\in J} \left| x_i\cdot y_j\cdot\langle \mv A_i,\mv A_j\rangle\right|\\
&\le \mu\cdot \sum_{i\in I}\sum_{j\in J} |x_i|\cdot|y_j|\\
&= \mu\cdot\left(\sum_{i\in I} |x_i|\right)\cdot\left(\sum_{j\in J} |y_j|\right)\\
&\le \mu k\cdot\|\mv x\|_2\cdot \|\mv y\|_2,
\end{align*}
where the last inequality follows from Cauchy-Schwarz. To complete the proof, we show that
\[\|\mv x\|_2,\|\mv y\|_2\le \frac{1}{\sqrt{1-\mu (k-1)}}.\]
We will bound $\|\mv x\|_2$: the bound on $\|\mv y\|_2$ follows from a similar argument. Consider the following relationships:
\begin{align*}
1&=\left\|\sum_{i\in I} x_i\cdot \mv A_i\right\|_2^2\\
&\ge \sum_{i\in I} x_i^2 -\sum_{i\neq i'\in I} \left| x_i\cdot x_{i'}\langle \mv A_i,\mv A_{i'}\rangle\right|\\
&\ge \|\mv x\|_2^2 -\mu \cdot \sum_{i\neq i'\in I} \left| x_i\cdot x_{i'}\right|\\
&= (1+\mu)\|\mv x\|_2^2 -\mu \cdot \sum_{i, i'\in I} \left| x_i\cdot x_{i'}\right|\\
&\ge (1+\mu)\|\mv x\|_2^2 -\mu \cdot \left(\sum_{i\in I} \left| x_i\right|\right)^2\\
&\ge (1+\mu)\|\mv x\|_2^2 -\mu k\|\mv x\|_2^2,
\end{align*}
where the last inequality follows from Cauchy-Schwarz. Rearranging the final inequality above (under our assumption that $\mu<1/(k-1)$) completes the proof.
\ep

\subsection{Proof of Corollary~\ref{prop:mu_k-upperbound2}}

\bp
If $\mu(\mv A)\geq \frac{1}{k-1}$, then $\frac{\mu_k(\mv A)}{\mu(\mv A)}\leq k-1\leq 2k-1$.

If $\mu(\mv A) < \frac{1}{k-1}$, then from Proposition~\ref{prop:mu_k-upperbound}:
\begin{align*}
\frac{\mu_k(\mv A)}{\mu(\mv A)}&\leq \min\left(\frac{k}{1-(k-1)\cdot\mu(\mv A)}, \frac{1}{\mu(\mv A)}\right)\\
&\leq \max_{x \in [0, 1]} \min\left(\frac{k}{1-(k-1) x}, \frac{1}{x}\right)
\end{align*}
$k/[1-(k-1) x]$ is increasing in $x$, while $1/x$ is decreasing in $x$. The optimal value for $x$ is the one that makes them equal.
$$kx=1-(k-1)x$$
$$x=\frac{1}{2k-1}.$$
\ep

\subsection{Proof of Proposition~\ref{prop:mu_k-lowerbound}}

\bp
If $\mu(\mv A)> \frac{1}{2k-1}$, then $\frac{k\cdot\mu(\mv A)}{1-(k-1)\cdot\mu(\mv A)}>1$.

For every integer $k>1$ and real number $c \geq 2k-1$, we will define a dictionary $\mv A$ whose coherence $\mu(\mv A)\leq \frac{1}{2k-1}$ and whose generalized coherence satisfies
\begin{equation}
\mu_k(\mv A) \geq \frac{k\cdot\mu(\mv A)}{1-(k-1)\cdot\mu(\mv A)}.
\label{eq:u_k-LB-condition}
\end{equation}
Let $\mv I_{k}$ denote the identity matrix of size $k \times k$, and $\mv J_{k}$ denote the all-ones matrix of size $k \times k$.
$$\mv A':=
\frac{1}{\sqrt{c^2+2k-1}}
\begin{bmatrix}
(c+1)\cdot\mv I_k-\mv J_k\\
\mv J_k
\end{bmatrix}
$$
$$\mv A'':=
\frac{1}{\sqrt{c^2+2k-1}}
\begin{bmatrix}
\mv J_k\\
(c+1)\cdot\mv I_k-\mv J_k
\end{bmatrix}
$$
$$\mv A:=
\begin{bmatrix}
\mv A' & \mv A''
\end{bmatrix}
$$
$\mv A$ is chosen so that the absolute value of the dot-product of any pair of columns is the same, no matter whether both columns come from $\mv A'$, both from $\mv A''$, or one from $\mv A'$ and the other from $\mv A''$.
\begin{equation}
\mu(\mv A)\quad=\quad\frac{2(c-k+1)}{c^2+2k-1}
\label{eq:u_k-LB-u}
\end{equation}
Notice that when $c$ ranges from $2k-1$ to $\infty$, $\mu(\mv A)$ ranges from $\frac{1}{2k-1}$ down to 0.

Let $\mv v$ be an all-ones vector of length $k$.
From \eqref{eq:defn-u_k} by choosing $I=\{1, \ldots, k\}$, $J=\{k+1, \ldots, 2k\}$, $\mv x=\frac{\mv v}{\|\mv A' \mv v\|}$, and $\mv y=\frac{\mv v}{\|\mv A'' \mv v\|}$, we get
\begin{align}
\mu_k(\mv A)\quad&\geq\quad\frac{\left|\langle\mv A' \mv v, \mv A'' \mv v\rangle\right|}
{\|\mv A' \mv v\|\cdot \|\mv A'' \mv v\|}\nonumber\\
&=\quad\frac{2k(c-k+1)}{(c-k+1)^2+k^2}
\label{eq:u_k-LB-u_k}
\end{align}
\begin{align}
1-(k-1)\mu(\mv A)\quad
&=\quad\frac{c^2+2k-1-2(k-1)(c-k+1)}{c^2+2k-1}\nonumber\\
&=\quad\frac{(c-k+1)^2+k^2}{c^2+2k-1}
\label{eq:u_k-LB-1-(k-1)u}
\end{align}
Combining \eqref{eq:u_k-LB-u}, \eqref{eq:u_k-LB-u_k}, and \eqref{eq:u_k-LB-1-(k-1)u}, we get \eqref{eq:u_k-LB-condition}.
\ep

\section{Missing Proofs from Section~\ref{sec:listsparse}}

\subsection{Proof of Proposition~\ref{prop:finite_list}}

\begin{proof}
        If $k<\spark{\mv A}$, then every $k$ columns of $\mv A$ are independent. Therefore, every subspace of $k$ columns gives a unique solution. The only way in which we can have multiple optimal solutions is having multiple subspaces returning equally-good solutions. However, the total number of subspaces is finite ${N \choose k}$.

If $k\geq\spark{\mv A}$, then there exist $k$ dependent columns. Any
vector $\mv b$ in the span of those $k$ columns can be represented with zero
error in infinitely many ways.
\end{proof}

\subsection{Proof of Lemma~\ref{lem:bad_b}}

\begin{proof}
Noting that the matrix $\mv A$ has $N \geq k+1$ columns, we consider
three cases.

First, if $\mv A$ has a set $T$ of $k$ linearly {\em dependent} columns,
then any vector $\mv b$ in the span of $T$ can be represented by columns in $T$
in infinitely many ways.

Second, suppose every subset of $k$ columns of $\mv A$ are independent,
but there is a subset $S$ of $k+1$ columns that are dependent.
In this case, every $k$-subset $T$ of $S$ spans {\em the same} $k$-dimensional
subspace, which is the span of $S$.
Let $\mv b$ be a vector in the span of $S$ but not in the span of any
$(k-1)$-subset of $S$. Such a vector $\mv b$ must exist because there are only
finitely many $(k-1)$-subsets of $S$, each of whose spans is of dimensionality
$k-1$, which is less than the dimensionality of $S$.
Now, let $T$ and $T'$ be two different $k$-subsets of $S$. It follows
that the representations of $\mv b$ using $T$ and $T'$ are different
because $|T\cap T'|=k-1$ and thus $\mv b$ cannot be represented using only
vectors in $T\cap T'$.
Consequently, $\mv b$ can be represented in at least $\binom{k+1}{k}$ different
ways.

Third, suppose every $k+1$ columns are independent. Consider the first $k+1$ columns of $\mv A$. Those columns (i.e., $\mv A_1, \ldots, \mv A_{k+1}$) can be viewed as points on the surface of a $(k+1)$-dimensional unit sphere. In spherical geometry, they form a simplex. Let $\mv b'$ be the center of the sphere inscribed in that simplex. $\mv b'$ is within the same distance from the span of any $k$ columns out of $\{\mv A_1, \ldots, \mv A_{k+1}\}$. Start with an infinitely small sphere that corresponds to the point $\mv A_1$. Expand this sphere while it remains contained in the simplex and tangent to $k$ of its facets (the $k$ facets that contain $\mv A_1$). Keep expanding this sphere until it hits the span of any subset $T$ of $k$ columns of $\mv A$. When it does, the center of the sphere $\mv b$ must have $k+1$ optimal representations. The sphere cannot keep expanding forever because eventually it will have to hit the facet opposite to $\mv A_1$, in which case $\mv b=\mv b'$. The idea is depicted in Figure \ref{fig balls} for $k=2$.
\end{proof}

\begin{figure}[ht!]
\begin{center}
\includegraphics[width=0.8\textwidth]{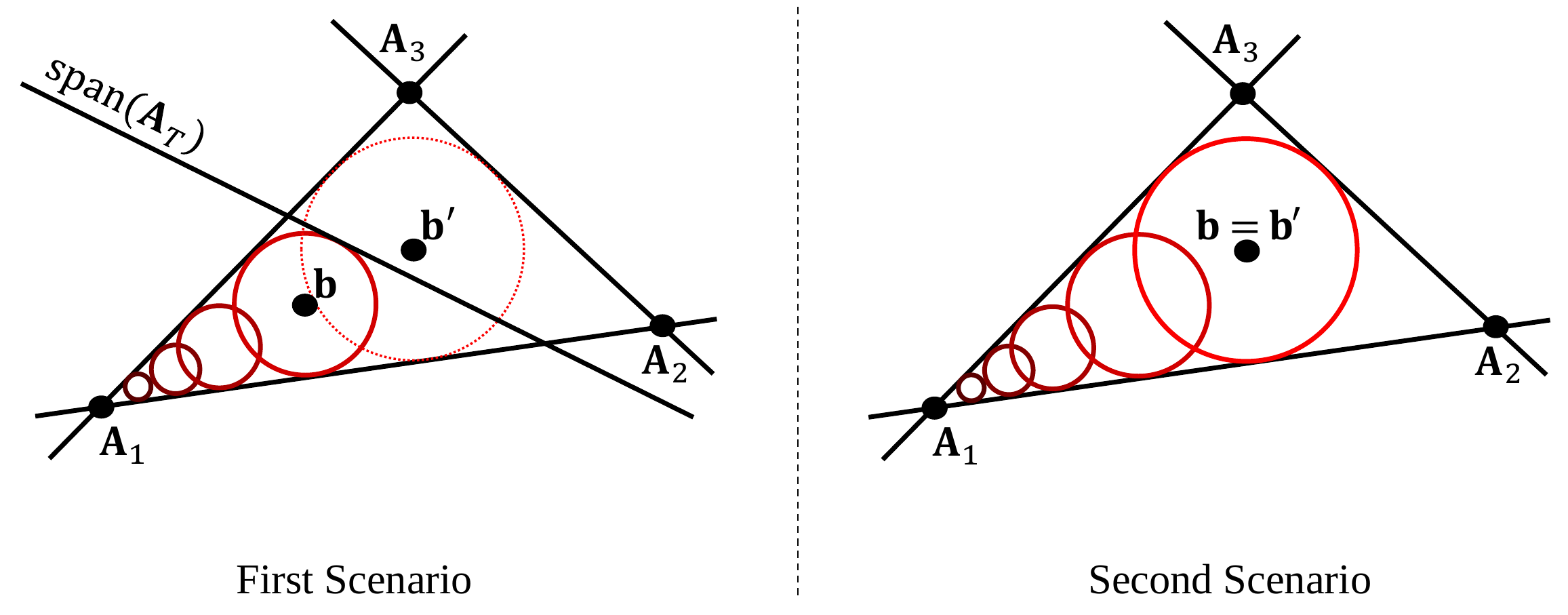}
\end{center}
\caption{The expanding sphere argument for $k=2$.}
\label{fig balls}
\end{figure}

\subsection{Proof of Proposition~\ref{prop:nec}}

\begin{proof}
Since $L$ is finite, from Proposition~\ref{prop:finite_list} we know
$k<\spark{\mv A}$. If $k=N$, then $\rank{\mv A} = N$.
If $k<N$, then from Lemma~\ref{lem:bad_b}, we have
\[ L \geq \max_{\mv b \neq \mv 0}s(\mv A, \mv b, k) \geq k+1. \]
 \end{proof}

\subsection{Proof of Proposition~\ref{prop:list-1}}

\begin{proof}
For the forward direction, suppose $\max_{\mv{\mv b \neq 0}} s(\mv A, \mv b,
k) = 1$. Then, from Proposition~\ref{prop:nec} either $k=N=\rank{\mv A}$ or
$k < \min\{1, \spark{\mv A}\}$. But $k \in [N]$; so it must hold that
$k=N=\rank{\mv A}$.

The backward direction is straightforward.
\end{proof}

\subsection{Proof of Proposition~\ref{prop:list-2}}

\begin{proof}
 When $L=2$, condition (\ref{eq:suf}) becomes
 $$k<\spark{\mv A} \wand \left[k=N \wor N=2\right],$$
 while condition (\ref{eq:nec}) becomes
 $$k<\spark{\mv A} \wand \left[k=N \wor k=1\right].$$
 They are still equivalent when $k\geq 2$.
 However, when $k=1$ and $N>2$, a gap emerges between the two conditions.

 When $k=1$ and $N>2$, the necessary and sufficient condition for having at
 most two solutions is that every two columns of $\mv A$ must be independent
 and every three columns must be dependent. This is equivalent to
 \begin{equation}\label{eq k=1}
 \rank{\mv A}=2 \wand \spark{\mv A}=3
 \end{equation}
 Combining (\ref{eq:suf}) with (\ref{eq k=1}), we get (\ref{eq:l<=2}).
 \end{proof}

\subsection{Proof of Proposition~\ref{prop:tight-k+1}}

\begin{proof}
From Lemma~\ref{lem:bad_b}, we have
$$\min_{\mv A \in \R^{m \times N}}\max_{\mv b\neq 0} s(\mv A, \mv b, k) \geq
k+1.$$
Hence, it is sufficient to construct a matrix $\mv A$ of dimension
$(k+1)\times N$ such that
\begin{equation}\label{eq:ub2}
\max_{\mv b\neq 0} s(\mv A, \mv b, k) \leq k+1.
\end{equation}
(For $m>k+1$, we simply pad the matrix with zeros.)
If $N=k+1$, we can construct $\mv A$ by choosing any arbitrary $k+1$ linearly
independent columns. The chosen $\mv A$ satisfies \eqref{eq:ub2} and has spark
$>k+1$. If $N>k+1$, then inductively we will assume that we have already
constructed a matrix $\mv A'$ of size $(k+1)\times(N-1)$ that satisfies
\eqref{eq:ub2} and has spark $>k+1$. We will construct $\mv A$ by adding an
$N$-th column to $\mv A'$ while maintaining \eqref{eq:ub2} and a spark $>k+1$.
We will start with choosing any arbitrary vector $\mv a_N$ (to be added to $\mv
A'$ to form $\mv A$) that does not belong to the span of any $k$ columns of $\mv A'$. (Such a vector $\mv a_N$ must exist because there is only a finite number of subspaces spanned by $k$ columns of $\mv A'$; each one of those subspaces is $k$-dimensional while the space has $k+1$ dimensions). By our initial choice of $\mv a_N$, $\spark{\mv A}>k+1$. We apply successive perturbation on $\mv a_N$ until $\mv A$ satisfies \eqref{eq:ub2}. Perturbations are infinitely small so that they don't interfere with $\spark{\mv A}>k+1$. Moreover, each perturbation is infinitely smaller than than the previous one such that it does not interfere with results of previous perturbations.

Because $\spark{\mv A}>k+1$, for every two subsets $S$ and $T$ of $k$ columns of $\mv A$, the spans of $S$ and $T$ are not identical unless $S$ and $T$ are identical. Given $k+1$ different subsets $S_1, \ldots, S_{k+1}$ of $k$ columns of $\mv A$, there are exactly $2^{k+1}$ unit vectors $\mv b$ that have the same distance to $\mspn{\mv A_{S_1}}, \ldots, \mspn{\mv A_{S_{k+1}}}$; one unit vector in each quadrant.

The successive perturbation of $\mv a_N$ goes as follows: Loop through all
possible choices of $k+1$ different subsets $S_1, \ldots, S_{k+1}$ of $k$
columns of $\mv A$ such that $\mv a_N$ appears in at least one subset (There are
finitely many such choices). For each such choice, loop through all unit vectors
$\mv b$ that have the same distance $d$ to $\mspn{\mv A_{S_1}}, \ldots,
\mspn{\mv A_{S_{k+1}}}$ (There are also finitely many such vectors). For each
such vector, loop through all possible choices of $k+1$ different subsets $T_1,
\ldots, T_{k+1}$ of $k$ columns of $\mv A$ such that each choice is different
from our earlier choice of $S_1, \ldots, S_{k+1}$. For each such choice, loop
through all unit vectors $\mv b'$ that have the same distance $d'$ to $\mspn{\mv
A_{T_1}}, \ldots, \mspn{\mv A_{T_{k+1}}}$. For each such vector, perturb $\mv
a_N$ such that either $d$ is different from $d'$ or $\mv b$ is different from
$\mv b'$. Each perturbation must be much smaller than the previous one such that
it does not interfere with effects of previous perturbations. After perturbation, no vector $\mv b$ has more than $k+1$ optimal solutions. Therefore, the perturbed $\mv A$ satisfies (\ref{eq:ub2}).
\end{proof}

\subsection{Proof of Proposition~\ref{prop:random-b}}

\begin{proof}
First, suppose $k<\spark{\mv A}-1$, i.e. every $k+1$ columns of $\mv
A$ are independent. This is equivalent to the fact that every set of
$k$ columns span a $k$-dimensional subspace that is different from the span of
any other choice of $k$ columns.
Because every set $S$ of $k$ columns are independent, if $\mv b$
is strictly closer to $\spn(S)$ than to $\spn(T)$ for every other subset $T$
of $k$ columns, then \sparse will have a unique solution.
Consider the Voronoi tessellation with the subspaces spanned by $k$-subsets
of columns of $\mv A$ as generators.
Because no two subspaces are identical, the boundary between the
Voronoi cells of any two subspaces has less dimensionality than the 
cells themselves. Therefore, the probability of $\mv b$ falling onto the
boundary is 0.

Conversely, consider the case when $k\geq\spark{\mv A}-1$. This means that
there exist $k+1$ dependent columns. Out of those, the span of some choice of
$k$ columns is identical to the span of some other choice of $k$ columns. The
two subspaces are identical. Therefore, they share an entire Voronoi cell that
has the same dimensionality as the whole space. The probability of $\mv b$
falling into this cell is not zero.
\end{proof}

\subsection{Proof of Proposition~\ref{prop:random-b-rank}}

\begin{proof}
A vector $\mv b$ has an infinite number of solutions if and only if it is within the Voronoi cell of some $k$ dependent columns of $\mv A$. The Voronoi cells of all choices of $k$ dependent columns are the only regions of the space associated with infinite numbers of solutions. Assume that $k\leq\rank{\mv A}$. This means that there exist $k$ independent columns of $\mv A$. If there exist $k$ dependent columns as well, then the span of those columns is properly contained in the span of some $k$ independent columns. Hence, the dimensionality of the Voronoi cell of the $k$ dependent columns is strictly less than the dimensionality of the Voronoi cell of the $k$ independent ones. Therefore, the probability of $\mv b$ falling into the former cell is zero.

Assume that $k>\rank{\mv A}$. This means that every $k$ columns of $\mv A$ are dependent. The Voronoi cells of all choices of $k$ dependent columns occupy the whole space. Therefore, with a probability of 1, the number of solutions would be infinite.
\end{proof}

\section{Missing Proofs from Section~\ref{sec:examples}}

\subsection{Proof of Lemma~\ref{lem:tight-eg}}

\bp
Let $m\ge 2$ be a large enough integer. We will define $\mv A^*$ to be the following $m\times m$ matrix. We will define the columns $\mv A_1^*,\dots,\mv A_m^*$ as follows: 
\[ \mv A_1^* =\frac{1}{\sqrt{m-1}}\cdot\sum_{j=2}^m \mv e_j,\]
and
for every $2\le i\le m$,
\[\mv A_i = \frac{1}{\sqrt{1+\eps^2 k}}\cdot \mv e_1 + \frac{\eps\sqrt{k}}{\sqrt{1+\eps^2k}}\cdot \mv e_i.\]

First, note that the coherence of $\mv A^*$ is given by
\[\max_{i\neq j>1} \langle\mv A_i^*,\mv e_1\rangle\cdot \langle\mv A_j^*,\mv e_1\rangle= \frac{1}{1+\eps^2k},\]
as required.

Finally, let $\vb=\mv e_1$. Then note that for every $\Lambda\subset \{2,\dots,m\}$ such that $|\Lambda|=k$, we have
\begin{align*}
\min_{\mv x\in\R^k}\|\mv A_{\Lambda}^*\mv x-\vb\|_2^2 &\le \left\|\sum_{j\in \Lambda} \frac{\sqrt{1+\eps^2k}}{k}\cdot \mv A_j^* -\mv e_1\right\|_2^2\\
&=\left\|\sum_{j\in \Lambda}\frac{1}{k}\cdot \mv e_1 + \sum_{j\in \Lambda} \frac{\eps}{\sqrt{k}}\cdot \mv e_j -\mv e_1\right\|_2^2 \\
& =\eps^2,
\end{align*}
as desired. Since there are at least $(m-1)/k$ choices of such disjoint $\Lambda$, the proof follows.
\ep

\subsection{$\Z_4$ Kerdock codes}
\label{app:kerdock}

Let $m > 3$ be an even integer and set $n = 2^m$. Let $K(m) \in \Z_4^n$ be the $\Z_4$ Kerdock code of length $n$. Consider the redundant dictionary $\mv A \in \C^n$ of size $n \times n^2$ given by ``exponentiating'' $K(m)$ codewords that are unique up to a scalar multiple of $\iconst = \sqrt{-1}$; i.e., each column of $\mv A$ corresponds to a (normalized) codeword $c_k \in K(m)$ and the $\ell$th row in that column corresponds to $\iconst$ raised to the power of the $\ell$the entry of $c_k$,
\[
	A_{\ell,k} = \frac{1}{\sqrt{n}} \iconst^{c_{\ell,k}}.
\]
From~\cite{CCKSKerdock1997}, we know that $\mv A$ is the union of $n$ orthonormal bases, $\mv A = [\mathcal{A}_1, \ldots, \mathcal{A}_n]$, with coherence $\mu(\mv A) = 1/\sqrt{n}$. For the fixed sparsity value $k = \sqrt{n}$, we create a set of vectors with a large list of different, disjoint sparse approximations. Let $\mv e_1$ be the first canonical basis vector in $\C^k$ (i.e., $\mv e_1$ has a `1' in the first coordinate, followed by $k-1$ `0's) and let $\mathbbm{1}$ be the vector of all `1's of length $n/k$. Construct the vector $\mv z = (\mathbbm{1} \otimes \mv e_1)$. Set
\[
	\mv x_i = [0, \ldots, 0, \mathcal{A}_i^{-1} \mv z, 0, \ldots, 0] \quad\text{for $i = 1,\ldots, n$.}
\]
It follows from~\cite{CalderbankSympleticKerdock:2007} that each of $\mv x_i$ has sparsity exactly $\sqrt{n}$.

We will argue the following:
\blmm
\label{lem:kerdock-lb}
For every integer $s\le n$, there exists a vector $\mv b$ such that
\begin{enumerate}
\item \label{item:num-solns} There are $\binom{n}{s}$ vectors $\mv x$ with distinct supports, each with sparsity $s\cdot \sqrt{n}$ such that $\mv A\mv x=\mv b$.
\item \label{item:overlap} Further in the set of solutions above, each atom in $\mv A$ appears in exactly $\frac{s}{n}$ fraction of the solutions.
\end{enumerate}
\elmm
\begin{proof}

Fix an $s\le n$ and define 
\[\mv b= s\cdot \mv z.\]
Now for every subset $S\subseteq [n]$ such that $|S|=s$, consider the vector
\[\mv x_S=\sum_{i\in S} \mv x_i.\]
Note that by definition of $\mv x_i$ and $\mv z$, we have that $\mv A\mv x_S=\mv b$. Further by definition of the vectors $\{\mv x_i\}_{i=1}^n$, we have that for distinct subsets $S$ and $T$, the vectors $\mv x_S$ and $\mv x_T$ have distinct supports. This implies that there are $\binom{n}{s}$ such vectors, which proves item~\ref{item:num-solns}.
Further, note that again by construction each column in $\mv A$ appears in exactly $\binom{n-1}{s-1}$ of these solutions, which in turn proves item~\ref{item:overlap}.
\end{proof}

Note that the above implies that for sparsity $k=s\cdot \sqrt{n}$, the coherence of the matrix satisfies $\mu(\mv A)=\frac{s}{k}$. We now use this example to illustrate the tightness of the bounds on $\mu$ in our upper bounds.

We begin with Corollary~\ref{cor:list-decoding-k>1}, which states that as long as $\mu<\frac{1}{2k-1}$, $L(\mv A,k,0,1)\le 1$. Note that Lemma~\ref{lem:kerdock-lb} (with $s=1$) implies that if we are allowed $\mu=\frac{1}{k}$, then we can have $L(\mv A,k,0,1)\ge n$. This implies that the bound of $\mu<\frac{1}{2k-1}$ in Corollary~\ref{cor:list-decoding-k>1} is necessary (and tight to within a factor of $2$).

Now consider Theorem~\ref{thm:gen-listapprox-ub}. Recall that we needed the condition $\mu\le O(1/k)$ to show that $L(\mv A,k,0,o(L))$ is $O(1)$. Now note that Lemma~\ref{lem:kerdock-lb} shows that this is necessary. In particular, as long as $s$ is $\omega(1)$ and at most $\sqrt{n}$, we have that $L(\mv A, k,0,r)$ can be super-polynomially large in $n$ under the condition that $\mu=\omega(1/k)$ even if we allow $r=o(L)$.\footnote{We only consider $s<\sqrt{n}$ since otherwise sparsity $k\ge n$, which is not an interesting regime.}


The Kerdock example above subsumes the next example of spikes and sines but we include it for completeness.

\subsection{Spikes and sines}
\label{app:spikes-and-sines}

Set $d \geq 2$, $n = 2^{2d}$, and consider the redundant dictionary $\mv A = [\mathcal{F},I]$  of size $n \times 2n$ that is the union of the Fourier basis and the canonical basis. Finally, set $k = 2^d = \sqrt{n}$. For a fixed sparsity value of $k$, we create a set of vectors with a list of different $k$-sparse approximations. Let $\mv e_1$ be the first canonical basis vector in $\R^k$ (i.e., $\mv e_1$ has a `1' in the first coordinate, followed by $k-1$ `0's) and let $\mathbbm{1}$ be the vector of all `1's of length $n/k$. Construct a vector $\mv v = (\mathbbm{1} \otimes \mv e_1)$. Set 
\[
	\mv z = \begin{bmatrix} \mv v \\ -{\mathcal F} \mv v \end{bmatrix}.
\]
It is clear that $\mv z$ is of length $2n$, has $2k$ nonzeros, and that $\mv A \mv z = \mv 0$ by construction. 

In an attempt to generate a long list of $k$-term representations, we set $\mv v^{(i)}(t) = \mv v(t-i)$ to be the $i$th circular shift of the vector $\mv v$ defined above. (Note that there can be at most $k$ circular shifts.) Then, set
\[
	\mv z^{(i)} = \begin{bmatrix} \mv v^{(i)} \\ -{\mathcal F} \mv v^{(i)} \end{bmatrix}.
\]
For each $i$, choose a set of size $k = \sqrt{n}$ from the support of $\mv z^{(i)}$ such that it has support exactly $k/2$ in $\mv v^{(i)}$ and call that set $\Omega_i$. Define $\mv x_1^{(i)} = \mv z|_{\Omega_i}$ and $\mv x_2^{(i)} = -\mv z + \mv x_1^{(i)}$ and construct
\[
	\mv x = \sum_{i}^k \mv x_1^{(i)}.
\]

Finally, define
\[\vb = \mv{A}\cdot\mv{x}.\]
We will argue the following:
\begin{lmm}
\label{lem:ss-lb}
There are $2^k$ solutions to the equation $\mv{A}\mv{x}=\vb$ each with sparsity $\Theta(k^2)$.
\end{lmm}
\begin{proof} We will construct $2^k$ solutions with the required sparsity. Indeed, we claim that any solution of the form $\sum_i \mv y^{(i)}$, where $\mv y^{(i)}\in \{x_1^{(i)}, x_2^{(i)}\}$ satisfies the required properties. Note that there are $2^k$ total such vectors. Next we formally define the vectors and argue their properties.

For every $\mv{c}\in \{0,1\}^k$, consider the vector
\[\mv y_{\mv{c}}= \sum_{i=1}^k \mv y^i_{\mv{c}},\]
where $\mv y^i_{\mv{c}}=\mv x_1^{(i)}$ if $\mv c_i=1$ and $\mv y^i_{\mv{c}}=\mv x_2^{(i)}$ otherwise.

First note that since by construction, $\mv{A} \mv x_1^{(i)}=\mv{A} \mv x_2^{(i)}$ for every $i\in [k]$, we have that for every $\mv c\in\{0,1\}^k$, $\mv{A} \mv y_{\mv c}=\vb$ as required.

Next, we argue that all $\mv y_{\mv{c}}$ are distinct and that they have the required sparsity. We note two properties:
\begin{enumerate}
\item[(i)] First note that for every $i\in [k]$, by definition, $\msupp{\mv v^{(i)}}$ are disjoint. 
\item[(ii)] Further, for any fixed $i\in [k]$, we have $\msupp{\mv x_1^{(i)}}\cap \msupp{\mv v^{(i)}}$ and $\msupp{\mv x_2^{(i)}}\cap \msupp{\mv v^{(i)}}$ are disjoint and have size exactly $k/2$. 
\end{enumerate}

These two facts immediately imply that all the vectors $\mv y_{\mv c}$ are distinct (and hence there are $2^k$ of them). Finally, since all of the $2k$ vectors $\mv x_1^{(i)}$ and $\mv x_1^{(i)}$ have sparsity exactly $k$, every $\mv y_{\mv c}$ has sparsity at most $k^2$. Further, the above two facts imply that each $\mv y_{\mv c}$ has sparsity at least $k^2/2$ (since it gets ``fresh" $k/2$ ones in $\msupp{\mv v^{(i)}}$ from each $\mv y^i_{\mv c}$).

\end{proof}

\end{document}